\documentclass{article}
\usepackage{graphicx} % Required for inserting images
\usepackage{amsmath}
\usepackage{amsfonts,amsmath,amsthm,amssymb}
\usepackage{mathtools}
\usepackage[numbers,sort]{natbib}
\usepackage{geometry}
\usepackage{pgf,pgfplots}
\usepackage{tikz}
\usepackage{caption}
\usepackage{subcaption}
\usepackage[]{enumitem}
\usepackage{bigints}
\usepackage{float}
\usepackage[title]{appendix}
\usepackage{url}
\usepackage{pdfpages}
\usepackage{enumitem}

\usepackage{bm}
\usepackage[normalem]{ulem}

\usepackage{xcolor}
\RequirePackage[colorlinks,citecolor=blue,urlcolor=blue]{hyperref}

\theoremstyle{plain}
\newtheorem{lem}{Lemma}

\newtheorem{thm}[lem]{Theorem}

\theoremstyle{definition}
\newtheorem{remark}[lem]{Remark}
\newtheorem{defn}[lem]{Definition}

\newcommand{\R}{\mathbb{R}}
\renewcommand{\P}{\mathbb{P}}

\newcommand{\E}{\mathbb{E}}
\newcommand{\F}{\mathcal{F}}
\newcommand{\N}{\mathbb{N}}

\newcommand{\Ccal}{{\mathcal C}}

\newcommand{\Fcal}{{\mathcal F}}

\newcommand{\Lcal}{{\mathcal L}}

\newcommand{\Ocal}{{\mathcal O}}
\newcommand{\Pcal}{{\mathcal P}}

\numberwithin{equation}{section}
\numberwithin{lem}{section}

\makeatletter
\renewenvironment{proof}[1][\proofname] {\par\pushQED{\qed}\normalfont\topsep6\p@\@plus6\p@\relax\trivlist\item[\hskip\labelsep\bfseries#1\@addpunct{.}]\ignorespaces}{\popQED\endtrivlist\@endpefalse}
\makeatother

\title{Calibrated rank volatility stabilized models for large equity markets\footnote{Acknowledgements: This work has been partially supported by the National Science Foundation under grant NSF DMS-2206062. The authors thank Johannes Ruf for helpful discussions regarding the data cleaning process and Steven Campbell for pointing out the unbiased volatility estimator used in this paper. %\eqref{eqn:sigma_estimator}.             
}}
\author{
David Itkin\footnote{Department of Mathematics, Imperial College London, \texttt{d.itkin@imperial.ac.uk}}
\and
Martin Larsson\footnote{Department of Mathematical Sciences, Carnegie Mellon University, \texttt{larsson@cmu.edu}}
}
\begin{document}

\maketitle

\begin{abstract}
    In the framework of stochastic portfolio theory we introduce rank volatility stabilized models for large equity markets over long time horizons. These models are rank-based extensions of the volatility stabilized models introduced in \cite{fernholz2005relative}. On the theoretical side we establish global existence of the model and ergodicity of the induced ranked market weights. We also derive explicit expressions for growth-optimal portfolios and show the existence of relative arbitrage with respect to the market portfolio. On the empirical side we calibrate the model to sixteen years of CRSP US equity data matching (i) rank-based volatilities, (ii) stock turnover as measured by market weight collisions, (iii) the average market rate of return and (iv) the capital distribution curve. Assessment of model fit and error analysis is conducted both in and out of sample. To the best of our knowledge this is the first model exhibiting relative arbitrage that has statistically been shown to have a good quantitative fit with the empirical features (i)-(iv). We additionally simulate trajectories of the calibrated model and compare them to historical trajectories, both in and out of sample. 
\end{abstract}

\paragraph*{Keywords:} Stochastic Portfolio Theory, Rank-Based Models, Equity Model Calibration, Relative Arbitrage.
	
	\paragraph*{MSC 2020 Classification:} 91G15, 62P05.

\section{Introduction}

Since the pioneering work of Markowitz \cite{markowitz1952portfolio}, many models have been proposed to capture important features of large equity markets. This is a challenging task especially when constructing models over long time horizons. For example, Markowitz's Modern Portfolio Theory framework requires estimating expected rates of return and covariances of individual assets, which is notoriously difficult to do given the low signal-to-noise ratios present in financial data. Additionally, the introduction of new securities and the delisting of others causes difficulties for calibration, which is especially prevalent over long-time horizons. Indeed, out of the largest 3500 stocks in the CRSP US equity universe on Jan 2, 1990 only 600 remained on December 31, 2022. 

Fernholz, in his monograph \cite{fernholz2002stochastic}, proposed the framework of Stochastic Portfolio Theory (SPT) as a descriptive theory of equity markets, which aims to only use observable/estimable quantities for equity modelling. One class of models that possesses desirable features for this goal are \emph{rank-based} models, such as first-order models \cite{fernholz2002stochastic,banner2005atlas,ichiba2011hybrid} or rank Jacobi models \cite{itkin2021open} which prescribe dynamics for the stocks in such a way that the evolution only depends on a stock's current rank. In numerous studies researchers have shown the stability of various rank-based properties prevalent in equity data, such as the distribution of capital and turnover in asset ranks, which can be used for calibration \cite{fernholz2002stochastic,fernholz2013second,banner2019diversification,campbell2022efficient}. Additionally, these rank-based models provide an elegant workaround to the aforementioned issue of new stocks entering/exiting the market. Indeed, suppose the stock currently occupying the 100th rank gets delisted. Then the stock occupying the 101st rank will take its place at rank 100. Since it is precisely the ranked asset, not the named one, which is relevant in these models the delisting event does not introduce difficulties in the same way as it would for classical name-based models.

The main contribution of this paper is the introduction of the \emph{rank volatility stabilized model} for equity markets, as well as its calibration to US equity data. The model generalizes the volatility stabilized model of \cite{fernholz2005relative} and is parsimonious in that each stock has two clearly interpretable parameters: a growth parameter and a volatility parameter. Despite only having two parameters per asset we are able to fit the following four features of long-term equity modelling:
\begin{enumerate}[noitemsep]
\item Quadratic variation of the ranked market capitalizations, \label{item:one}
\item \label{item:two} The asset turnover at each rank as measured by the collision local times of the market weights, %The turnover of assets in the market as measured by asset collisions when stocks switch rank, 
\item The annual rate of return for the entire market, \label{item:three}
\item The capital distribution curve. \label{item:four}
\end{enumerate}
The precise definitions of these notions are given in Section~\ref{sec:calibration}. We calibrate the model to daily CRSP equity data for the 16-year period  Jan 2, 1990 - Dec 31, 2005 and compare the performance of the model to the 17-year out of sample period from Jan 3, 2006 - Dec 31, 2022.

Additionally, we show that the rank volatility stabilized model admits \emph{relative arbitrage} with respect to the market portfolio. That is, there exists a time horizon, which depends on the model parameters, and an independent long-only trading strategy that, in the model, will outperform the market portfolio over the time horizon with probability one. The existence of relative arbitrage (in the absence of trading frictions, which we don't consider in this study) has been established under mild conditions in the literature \cite{fernholz2002stochastic,fernholz2005relative,larsson2021relative,banner2008short}. Nevertheless, prior to this work, we are unaware of any model that was calibrated to fit the criteria \ref{item:one}-\ref{item:four} and admits relative arbitrage. 

Although we do not claim that relative arbitrage exists in real equity markets, it is striking that there exist models consistent with these empirical features and relative arbitrage. The portfolios we consider in this study that achieve relative arbitrage in this model are a one-parameter long-only family of portfolios called \emph{diversity-weighted} portfolios. In contrast to the growth-optimal portfolios in this model -- which we derive explicit formulas for in both the \emph{closed market} setup as well as in the \emph{open market} setup recently proposed in \cite{fernholz2018numeraire,karatzas2020open}, where investment is constrained to a subset of the market consisting of the largest assets -- the portfolios achieving relative arbitrage do not require leverage, can be implemented by unsophisticated investors and have been shown to perform well on real data over long time horizons \cite{ruf2020impact,campbell2021functional}. This highlights the potential for models of this type to inform portfolio construction. A detailed analysis of this type is beyond the scope of this paper, but in Section~\ref{sec:portfolio_discussion} we summarize possible approaches to use our calibrated model for portfolio selection. 

The paper is organized as follows. In Section~\ref{sec:model} we introduce the rank volatility stabilized model and establish several theoretical properties it possesses, such as existence of the process, uniqueness of its associated reflected stochastic differential equation and ergodicity of the ranked market weights. Section~\ref{sec:calibration} then develops estimators for the model and calibrates them to items \ref{item:one}-\ref{item:four}. An assessment of the model fit, including error analysis and comparison of simulated trajectories is done in Section~\ref{sec:model_performance}. Section~\ref{sec:portfolio_main} then considers portfolio optimization in the rank volatility stabilized model. Growth-optimal strategies are derived in this section along with the existence of relative arbitrage. Proofs of theoretical results and certain derivations are contained in Appendices~\ref{sec:proofs} and \ref{app:ranked_estimators} respectively for better readability.

% \begin{itemize}
% \item Discuss motivation of fitting the three quantities: (i) volatilities, (ii) collisions, (iii) CDC. 
% \item Compare to literature, most notably Fernholz and first order models.
% \item Mention the different features of this model such as nonconstant (log)-volatility and the existence of relative arbitrage.
% \item Also can discuss the mathematical properties established, such as boundary nonattainment, which can be of independent interest.
% \end{itemize}
\paragraph*{\textbf{Notation}}
\begin{itemize}
\item For $d \in \N$ we denote by $\R^d_+$ the set of $d$-dimensional vectors with nonnegative components. $\R^d_{++}$ denotes all such vectors with strictly positive components.
\item We denote the $d$-dimensional simplex by \[\Delta^{d-1} := \{x \in \R^d_+: x_1 + \dots + x_d = 1\}.\]
We also define its interior $\Delta^{d-1}_+ := \Delta^{d-1} \cap \R^d_{++}$. Analogously, we define the ordered simplex 
\[\nabla^{d-1} := \{y \in \Delta^{d-1}: y_1 \geq \dots \geq y_d\}\] and its interior $\nabla^{d-1}_+ := \nabla^{d-1} \cap \R^d_{++}$.
\item  For a vector $x \in \R^d$ we write $x_{()} = (x_{(1)},\dots,x_{(d)})$ for its ranked vector, which is a permutation of $x$ that satisfies $x_{(1)} \ge x_{(2)} \ge \dots \ge x_{(d)}$.
 We also define the rank identifying functions $r_i: \R^d \to \{1,\dots,d\}$ for $i = 1,\dots,d$ as well as the
name identifying function $n_k: \R^d
\to \{1, \dots , d\}$ for $k = 1, \dots, d$  via 
\begin{align*}
    r_i(x) & = k, \text{ where $k$ is such that } x_i = x_{(k)} \\
    n_k(x) & = i \text{ where $i$ is such that } x_i = x_{(k)}
\end{align*}and ties are broken by lexicographical ordering.
\item For a $d$-dimensional stochastic process $X$ with nonnegative components we set \begin{equation} \label{eqn:tau}
    \tau^X := \inf\{t \geq 0: X_i(t) = 0 \text{ for some } i = 1,\dots,d\} = \inf\{t \geq 0: X_{(d)}(t) = 0\}
\end{equation}
to be the first hitting time of zero for any component of $X$.
%\item For a vector $x \in \R^d$ and $k \in \{1,\dots,d\}$ we write $\overline x_k := x_k + \dots + x_d$ and $x_k^* := \max\{x_k,\dots,x_d\}$. We simply write $\overline x$ for $x_1 + \dots + x_d$.
\end{itemize}

\section{Rank volatility stabilized models} \label{sec:model}

We consider a fixed number of $d \ge 2$ stocks with strictly positive capitalizations $S_i(t)$, $i=1,\ldots,d$. The \emph{market weights} of the stocks are defined by
\[
X_i(t) = \frac{S_i(t)}{\overline S(t)}, \quad i = 1,\ldots,d,
\]
where $\overline S(t) = S_1(t) + \cdots + S_d(t)$ is the total value of the entire market. In the \emph{rank volatility stabilized model} the capitalizations evolve according to
\begin{equation} \label{eqn:S_dynamics}
	dS_i(t) = a_{r_i(t)}\overline S(t)\, dt  + \sigma_{r_i(t)}\sqrt{S_i(t)\overline S(t)}\, dW_i(t), \quad i = 1,\dots,d,
%	dS_i(t) = a_{r_i}\overline S(t)\, dt  + \sigma_{r_i}\sqrt{S_i(t)\overline S(t)}\, dW_i(t), \quad i = 1,\dots,d,
\end{equation}
for some growth parameters $a_k \in \R$ and volatility parameters $\sigma_k > 0$, $k=1,\ldots,d$. Here $W_1,\ldots,W_d$ are independent standard Brownian motions and $r_i(t)$ is shorthand for $r_i(S_1(t),\ldots,S_d(t))$, the rank of the $i$th stock at time $t$. Equivalently, the dynamics of the stock returns are
\begin{equation} \label{eqn:dS_by_S_dynamics}
\frac{dS_i(t)}{S_i(t)} = \frac{a_{r_i(t)}}{X_i(t)}\, dt + \frac{\sigma_{r_i(t)}}{\sqrt{X_i(t)}}\, dW_i(t), \qquad i=1,\dots,d.
\end{equation}
Summing \eqref{eqn:S_dynamics} over $k$ shows that the dynamics of the overall market return are
\begin{equation} \label{eq_Sbar_returns}
\frac{d\overline S(t)}{\overline S(t)} = \lambda\, dt+ \sum_{i=1}^d \sigma_{r_i(t)}\sqrt{X_i(t)}\, dW_i(t) %= \overline a\, dt + \sqrt{\sum_{k=1}^d \sigma_k^2 X_{(k)}}\, d\overline W(t),
%\frac{d\overline S(t)}{\overline S(t)} = \bar a\, dt+ \sum_{i=1}^d \sigma_{r_i(t)}\sqrt{X_i(t)}\, dW_i(t) %= \overline a\, dt + \sqrt{\sum_{k=1}^d \sigma_k^2 X_{(k)}}\, d\overline W(t),
%\frac{d\overline S(t)}{\overline S(t)} = \bar a\, dt+ \sum_{k=1}^d \sigma_k\sqrt{X_{(k)}(t)}\, dW_k(t) %= \overline a\, dt + \sqrt{\sum_{k=1}^d \sigma_k^2 X_{(k)}}\, d\overline W(t),
\end{equation}
where
\[
\lambda = a_1 + \cdots + a_d.
\]
Thus, while the individual stocks have stochastic rates of return $a_{r_i(t)} / X_i(t)$, the overall market has a constant rate of return $\lambda$. The spot variance of the overall market is however stochastic in general and given by $\sum_{k=1}^d \sigma_k^2 X_{(k)}(t)$, where we recall that $X_{(k)}(t)$, $k=1,\ldots,d$, are the market weights in decreasing order.  Finally, the dynamics of the market weights are
%where we recall that $\bar a = a_1 + \cdots + a_d$. Thus, while the individual stocks have stochastic rates of return $a_{r_i(t)} / X_i(t)$, the overall market has a constant rate of return $\bar a$. The spot variance of the overall market is however stochastic in general and given by $\sum_{k=1}^d \sigma_k^2 X_{(k)}(t)$, where we recall that $X_{(k)}(t)$, $k=1,\ldots,d$, are the market weights in decreasing order.  Finally, the dynamics of the market weights are
\begin{equation} \label{eqn:X_dynamics}
\begin{split}
	dX_i(t) & = \left(a_{r_i(t)} - \lambda X_i(t) - \sigma^2_{r_i(t)}X_i(t) + X_i(t)\sum_{j=1}^d \sigma_{r_j(t)}^2X_j(t)\right)dt\\
	& \hspace{1cm} + \sigma_{r_i(t)}\sqrt{X_i(t)}\, dW_i(t) - X_i(t)\sum_{j=1}^d \sigma_{r_j(t)}\sqrt{X_j(t)}\, dW_j(t)
%	dX_i(t) & = \left(a_{r_i(t)} - \overline aX_i(t) - \sigma^2_{r_i(t)}X_i(t) + X_i(t)\sum_{j=1}^d \sigma_{r_j(t)}^2X_j(t)\right)dt\\
%	& \hspace{1cm} + \sigma_{r_i(t)}\sqrt{X_i(t)}\, dW_i(t) - X_i(t)\sum_{j=1}^d \sigma_{r_j(t)}\sqrt{X_j(t)}\, dW_j(t)
\end{split}
\end{equation}  
for $i=1,\dots,d$. Because $r_i(t) = r_i(X_1(t),\ldots,X_d(t))$, the market weights satisfy an autonomous SDE.

\begin{remark}
The classical volatility stabilized model of \cite{fernholz2005relative} is obtained by letting both the drift and volatility parameter vectors be constant (i.e., $a_k=a_l$ and $\sigma_k = \sigma_l$ for all $k,l$), while if only the volatility vector is constant one obtains the rank Jacobi model of \cite{itkin2021open} for the market weights. The name volatility stabilized refers to the fact that each individual asset's dynamics are stabilized by their market weight in order to produce long-term stability of the market weight vector.
\end{remark}

We now establish existence of processes $S = (S_1,\ldots,S_d)$ and $X = (X_1,\ldots,X_d)$ satisfying \eqref{eqn:S_dynamics} and \eqref{eqn:X_dynamics}. This is not immediate due to the discontinuous volatility coefficients and non-uniform ellipticity of the volatility matrix. Nevertheless, we have the following result, where $\tau^X$ denotes the first time a component of $X$ hits zero, and similarly for $\tau^S$; see \eqref{eqn:tau}. The proof is in Appendix~\ref{app:pf_existence_no_blowup}.

\begin{thm}[Existence] \label{thm:existence}
For any initial condition $s^0 \in \R^d_{++}$ there exists a weak solution $S$ to \eqref{eqn:S_dynamics} on the stochastic time interval $[0,\tau^S)$ with $S(0) = s^0$. Similarly, for any $x^0 \in \Delta^{d-1}_+$ there exists a weak solution $X$ to \eqref{eqn:X_dynamics} on $[0,\tau^X)$ with $X(0) = x^0$.
\end{thm}

Since Theorem~\ref{thm:existence} only asserts existence before the first hitting time of zero, it is of interest to establish a condition on the parameters that guarantees this does not occur. This is done in the following result, whose proof is in Appendix~\ref{app:pf_existence_no_blowup}.

\begin{thm}[Boundary non-attainment]\label{prop:no_blowup} 
Assume the Feller-type condition
\begin{equation} \label{eqn:parameter_condition}
a_k + \cdots + a_d \geq \max\left\{ \frac{\sigma_k^2}{2}, \ldots, \frac{\sigma_d^2}{2}\right\}, \quad k = 2,\dots,d.
%a_k + \cdots + a_d \geq \frac{1}{2} \max\{ \sigma_k^2, \ldots, \sigma_d^2\}, \quad k = 2,\dots,d.
\end{equation} 
Then any solution $S$ to \eqref{eqn:S_dynamics} and $X$ to \eqref{eqn:X_dynamics} satisfies $\tau^S = \tau^X = \infty$ almost surely. Consequently, under \eqref{eqn:parameter_condition} we have global existence for \eqref{eqn:S_dynamics} and \eqref{eqn:X_dynamics}.
\end{thm}

\begin{remark}
If all $\sigma_k$ are equal, the condition \eqref{eqn:parameter_condition} is also known to be necessary for $X$ to avoid the boundary; see \cite{itkin2021open}.
\end{remark}

The calibration method developed in later sections depends on properties of the ranked market weight process $X_{()} = (X_{(1)}, \ldots, X_{(d)})$, which we now investigate. The formulas of \cite{Banner2008Local} yield the dynamics
 \begin{equation} \label{eqn:ranked_weights}
	\begin{split}
		dX_{(k)}(t) & = \left(a_k - \lambda X_{(k)}(t) - \sigma^2_{k}X_{(k)}(t) + X_{(k)}(t)\sum_{j=1}^d \sigma_{j}^2X_{(j)}(t)\right)dt\\
		& \hspace{1cm} + \sigma_{k}\sqrt{X_{(k)}(t)}\, dB_k(t) - X_{(k)}(t)\sum_{j=1}^d \sigma_{j}\sqrt{X_{(j)}(t)}\, dB_j(t) + d\Phi_k(t)
%		dX_{(k)}(t) & = \left(a_k - \overline aX_{(k)}(t) - \sigma^2_{k}X_{(k)}(t) + X_{(k)}(t)\sum_{j=1}^d \sigma_{j}^2X_{(j)}(t)\right)dt\\
%		& \hspace{1cm} + \sigma_{k}\sqrt{X_{(k)}(t)}\, dB_k(t) - X_{(k)}(t)\sum_{j=1}^d \sigma_{j}\sqrt{X_{(j)}(t)}\, dB_j(t) + d\Phi_k(t)
	\end{split}
\end{equation}
for $k = 1,\dots,d$. Here $B_k(t) = \sum_{i=1}^d \int_0^t 1_{\{r_i(s) = k\}}\, dW_i(s)$ are independent standard Brownian motions, and the boundary reflection terms
\begin{equation} \label{eqn:Phi}
	d\Phi_k(t) = \frac{1}{N_k(t)}\left(\sum_{j=k+1}^d dL_{k,j}(t)- \sum_{j=1}^{k-1}\,  dL_{j,k}(t)\right)
\end{equation}
are given in terms of $L_{j,k}$, the local time at zero of $X_{(j)} - X_{(k)}$, and $N_k(t) = |\{i: X_i(t) = X_{(k)}(t)\}|$, the number of weights that occupy the $k$th rank at time $t$. This shows that $X_{()}$ satisfies a \emph{reflected stochastic differential equation} (RSDE), and we assert that this equation is well-posed. The basic theory of RSDEs is reviewed in Appendix~\ref{app:rsde}, where the following theorem is also proved.

\begin{thm}[RSDE well-posedness]\label{thm:rsde}
	Consider the  RSDE on $\nabla^{d-1}_+$ with normal reflection
\begin{equation} \label{eqn:Y_rsde}
	\begin{split}
		dY_k(t) & = \left(a_k - \lambda Y_{k}(t) - \sigma^2_{k}Y_{k}(t) + Y_{k}(t)\sum_{j=1}^d \sigma_{j}^2Y_{j}(t)\right)dt\\
		& \hspace{1cm} + \sigma_{k}\sqrt{Y_{k}(t)}\, dB_k(t) - Y_{k}(t)\sum_{j=1}^d \sigma_{j}\sqrt{Y_{j}(t)}\, dB_j(t) + d\Phi_k(t)
%		dY_k(t) & = \left(a_k - \overline aY_{k}(t) - \sigma^2_{k}Y_{k}(t) + Y_{k}(t)\sum_{j=1}^d \sigma_{j}^2Y_{j}(t)\right)dt\\
%		& \hspace{1cm} + \sigma_{k}\sqrt{Y_{k}(t)}\, dB_k(t) - Y_{k}(t)\sum_{j=1}^d \sigma_{j}\sqrt{Y_{j}(t)}\, dB_j(t) + d\Phi_k(t)
	\end{split} 
\end{equation} 
for $k=1,\dots,d$. There exists a pathwise unique solution $(Y,\Phi)$ to \eqref{eqn:Y_rsde} on the stochastic interval $[0,\tau^Y)$. In particular, the representation $Y = X_{()}$, where $X$ is any solution to \eqref{eqn:X_dynamics} and $\Phi$ is given by \eqref{eqn:Phi}, holds on $[0,\tau^Y) = [0,\tau^X)$. As such, under the parameter restriction \eqref{eqn:parameter_condition} we have global well-posedness for \eqref{eqn:Y_rsde}, and $X_{()}$ is a strong Markov process.
\end{thm}

In addition to the dynamics \eqref{eqn:ranked_weights}, our calibration method relies crucially on ergodicity of the ranked market weight process. Suppose that \eqref{eqn:parameter_condition} holds. From Theorem~\ref{thm:rsde}, we have that $X_{()}$ is a strong Markov process. By viewing it as a process on the larger compact state space $\nabla^{d-1}$ we immediately obtain the existence of an invariant measure on $\nabla^{d-1}$. The non-attainment of the boundary implies that the measure is supported on $\nabla^{d-1}_+$ and uniform ellipticity of the diffusion matrix on any compact set $K \subset \nabla^{d-1}_+$ can then be used to establish uniqueness of the invariant measure using an approach developed in \cite{harrison1987brownian} for reflected Brownian motion. This leads us to the following ergodicity result, whose detailed proof is given in Appendix~\ref{app:ergodic}. We denote by $\P_y$ the law of $X_{()}$ when initiated at $X_{()}(0) = y \in \nabla^{d-1}_+$. 

\begin{thm}[Ergodicity] \label{thm:ergodicity}
Under the condition \eqref{eqn:parameter_condition}, $X_{()}$ admits a unique invariant measure $\nu$ on $\nabla^{d-1}_+$ and we have the Birkhoff ergodic theorem
\[\lim_{T \to \infty} \frac{1}{T}\int_0^T f(X_{()}(t))\, dt = \int_{\nabla^{d-1}_+} f(y)\, d\nu(y), \qquad \P_y\text{-a.s.,}\]
    for every $\nu$-integrable function $f$ and every $y \in \nabla^{d-1}_+$.
\end{thm}

\begin{remark}
The uniqueness result in Theorem~\ref{thm:rsde} ensures that any statistic of the ranked market weights has a distribution that is uniquely determined by the parameters of the model. This is clearly a desirable feature. Note however that this is not the same as statistical identifiability of the model parameters. Indeed, the calibration procedure developed below would have recovered the parameters uniquely from data (up to statistical errors) even in the absence of uniqueness for \eqref{eqn:Y_rsde}, at least if the ergodicity property in Theorem~\ref{thm:ergodicity} could be established without uniqueness. Furthermore, let us emphasize that while we do establish uniqueness for \eqref{eqn:Y_rsde}, uniqueness for \eqref{eqn:S_dynamics} and \eqref{eqn:X_dynamics} remains an open question, the key difficulty being the discontinuous nature of the volatility process. % However, as established in Theorem~\ref{thm:rsde} below, uniqueness does hold for the \emph{reflected} SDE that the \emph{ranked} market weight process satisfies.
\end{remark}

\section{Data and calibration} \label{sec:calibration}

In this section we discuss how the parameters $a_k$ and $\sigma_k$ specifying the rank volatility stabilized model can be calibrated. We aim to fit the four features \ref{item:one}--\ref{item:four} laid out in the introduction. The first feature \ref{item:one} pins down the volatility parameters $\sigma_k$, while the second feature \ref{item:two} can be used to determine all but one of the growth parameters $a_k$. The one remaining degree of freedom is fixed by selecting $\lambda = a_1+\cdots+a_d$, which is the market rate of return; see \eqref{eq_Sbar_returns}. We establish a relationship showing that $\lambda$ offers a trade-off between closely fitting the capital distribution curve on one hand, and the collisions, which measure market turnover, on the other hand. Indeed, increasing $\lambda$ improves the fit to the curve, but at the expense of the fit to the collisions. Remarkably, the value $\lambda = 0.11$ that matches the historical market rate of return also turns out to offer a balanced fit to both the empirical collisions and capital distribution curves. We view this as our preferred calibrated model, though later on in Subsection~\ref{sec:model_performance_summary} we also discuss some benefits and drawbacks of using other values of $\lambda$.

We calibrate the model to US equity data for the 16-year period from Jan 1, 1990 to Dec 31, 2005. We use daily price data from  the CRSP equity universe. As usual, we filter equities to only include common stock. The raw CRSP data is cleaned and preprocessed according to the publicly available code provided by Ruf \cite{ruf2023github}. Then, for each trading day in the period we take the $d$ largest stocks as determined by market capitalization's from the previous day. In the analysis to come we set $d = 3500$, as the cleaned equity universe consists of at least this number of equities for the duration of the estimation period. All the parameters are calibrated using this dataset. In Section~\ref{sec:model_performance} we also assess the performance of the model out-of-sample using analogously processed data for the time period Jan 1, 2006 to Dec 31, 2022.

\subsection{Volatility calibration} \label{S_vol_calibration}

The dynamics \eqref{eqn:dS_by_S_dynamics} imply that the quadratic variations of the ranked log-capitalizations satisfy
\begin{equation} \label{eqn:cap_QV}
\frac{1}{T}\langle \log S_{(k)} \rangle(T) = \frac{\sigma_k^2}{T}\int_0^T \frac{1}{X_{(k)}(t)}\, dt, \qquad k=1,\dots,d.
\end{equation}
Theorem~\ref{thm:ergodicity} shows that, in the model, the right-hand side of \eqref{eqn:cap_QV} stabilizes for large $T$. This is also empirically supported by the stability of the capital distribution curve. Given observation times $0 = t_0 < t_1 < \dots < t_N = T$ we discretize the integral and replace the quadratic variation by a sum of squared increments to obtain the estimator
\begin{equation} \label{eqn:sigma_estimator}
\frac{\sum_{i=0}^{N-1} (\log S_{n_k(t_i)}(t_{i+1}) - \log S_{n_k(t_i)}(t_i))^2}{\sum_{i=0}^{N-1}{X_{(k)}^{-1}(t_i)}\Delta t_i}, \quad k = 1,\ldots,d,
%\widehat \sigma_k^2 := \frac{\sum_{i=0}^{N-1} (\log S_{n_k(t_i)}(t_{i+1}) - \log S_{n_k(t_i)}(t_i))^2}{\sum_{i=0}^{N-1}{X_{(k)}^{-1}(t_i)}\Delta t_i}, \quad k = 1,\ldots,d,
\end{equation}
for the volatility vector, where $\Delta t_i = t_{i+1} - t_i$. Since we use daily observations we have $\Delta t_i = 1/252$. Using the estimator \eqref{eqn:sigma_estimator} directly yields volatility parameters that generally decrease with rank but are noisy, which is undesirable for simulating the model. For this reason we smooth the vector in \eqref{eqn:sigma_estimator}. We use uniform averaging with a window size of 15, although the precise choice of filter does not have a material impact on the results. We denote the smoothed vector by $(\widehat\sigma_1^2, \ldots, \widehat\sigma_d^2)$ and use this as our estimate for $(\sigma_1^2,\ldots,\sigma_d^2)$. Figure~\ref{fig:sigma} shows the ``raw'' estimates \eqref{eqn:sigma_estimator} (thin gray line) as well as the processed estimates $\widehat\sigma_1^2, \ldots, \widehat\sigma_d^2$ (thick solid line).

\begin{figure}
\begin{center}
\includegraphics[scale=0.7]{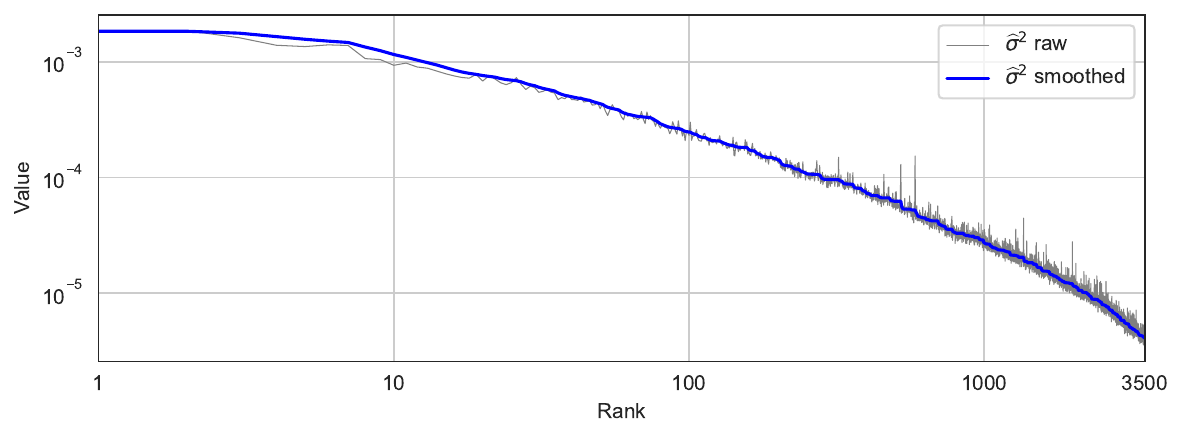}
\caption{Volatility parameter estimates $\widehat\sigma_k^2$, $k=1,\ldots,d$ (thick solid line), obtained using uniform averaging of the raw estimates in \eqref{eqn:sigma_estimator} (thin gray line).}
\label{fig:sigma}
\end{center}
\end{figure}

\begin{remark}
Empirically, the spot variance of the $k$th ranked stock tends to increase with $k$ but remain within the same order of magnitude for different values of $k$. In view of \eqref{eqn:dS_by_S_dynamics}, the model spot variance is $\sigma_k^2 / X_{(k)}(t)$. Therefore, because the market weights $X_{(k)}(t)$ have orders of magnitude that decrease with $k$, so must the volatility parameters $\sigma_k^2$. This explains the decreasing shape in Figure~\ref{fig:sigma}.
%Figure~\ref{fig:sigma} plots a smoothed version of $\widehat \sigma^2$ on a log-log scale, by first performing an isotonic regression for monotonicity and then applying an exponential smoother to remove estimation noise. Unsurprisingly, $\widehat \sigma^2$ is decreasing with rank. Indeed, although smaller stocks tend to be more volatile, volatility levels are of the same order of magnitude for large and small stocks. However, the inverse market weight $1/\sqrt{X_{(k)}}$ appears in the volatility of the capitalization dynamics \eqref{eqn:S_dynamics} and the order of magnitude of this quantity increases with $k$. Hence, one needs $\sigma$ to scale accordingly so that the ratio $\sigma_k^2/X_{(k)}$ has the same order of magnitude across different ranks $k$. This leads to the decreasing relationship we see in \eqref{fig:sigma}. 
\end{remark}

\begin{remark} \label{rem_vol_bias}
In the numerator of \eqref{eqn:sigma_estimator} one could instead use increments of the form
\begin{equation} \label{eqn:rank_dif}
    \log S_{(k)}(t_{i+1}) - \log S_{(k)}(t_i).
\end{equation}
This differs from \eqref{eqn:sigma_estimator} whenever the $k$th ranked stock switches name from time $t_i$ to $t_{i+1}$. The estimator in \eqref{eqn:sigma_estimator} turns out to perform better, as the one with \eqref{eqn:rank_dif} in the numerator exhibits bias. We discuss this point further in Appendix~\ref{app:vol_estimator}.
\end{remark}

\subsection{Collision calibration}  \label{S_collision_calibration}

Next we examine the collision local times, which will serve as a stepping stone toward estimating the growth parameters $a_k$. Specifically, we are interested in obtaining estimates for the quantities
\begin{equation} \label{eq_phi_k_def}
\phi_k := \lim_{T \to \infty} \frac{\Phi_k(T)}{T}, \quad k = 1,\dots, d,
\end{equation}
where $\Phi_k$ is the boundary reflection process in \eqref{eqn:Phi}. The ergodic property of Theorem~\ref{thm:ergodicity} guarantees that $\phi_k$ exists.
%Note also that $\sum_{k=1}^d \phi_k = 0$.  
To estimate $\phi_k$ we first consider the sums
\begin{equation} \label{eqn:phi_bar_def}
\overline\phi_k = \phi_1 + \cdots + \phi_k, \quad k=1,\ldots,d,
\end{equation}
where we note that $\overline\phi_d = \phi_1 + \cdots + \phi_d = 0$. An estimator for $\overline \phi_k$ is
\begin{equation} \label{eqn:hat_bar_phi}
\widehat{\overline\phi}_k = \frac{1}{T} \sum_{i=0}^{N-1}(X_{n_1(t_i)}(t_i) + \dots + X_{n_{k}(t_i)}(t_i))\log\left( \frac{S_{n_1(t_{i+1})}(t_{i+1}) + \dots + S_{n_{k}(t_{i+1})}(t_{i+1})}{S_{n_1(t_{i})}(t_{i+1}) + \dots + S_{n_{k}(t_{i})}(t_{i+1})}\right)
%-\frac{1}{T}\sum_{i=0}^{N-1}(X_{n_1(t_i)}(t_i) + \dots + X_{n_k(t_i)}(t_i))\log\left( \frac{S_{n_1(t_{i+1})}(t_{i+1}) + \dots + S_{n_k(t_{i+1})}(t_{i+1})}{S_{n_1(t_{i})}(t_{i+1}) + \dots + S_{n_k(t_{i})}(t_{i+1})}\right)
%    \widehat {\overline{\phi}}_k := -\frac{1}{T}\sum_{i=0}^{N-1}(X_{n_1(t_i)}(t_i) + \dots + X_{n_k(t_i)}(t_i))\log\left( \frac{S_{n_1(t_{i+1})}(t_{i+1}) + \dots + S_{n_k(t_{i+1})}(t_{i+1})}{S_{n_1(t_{i})}(t_{i+1}) + \dots + S_{n_k(t_{i})}(t_{i+1})}\right)
\end{equation}
for a large time horizon $T$. This estimator is inspired by one developed by Fernholz \cite[Equation~(5.4.1)]{fernholz2002stochastic} for a different, but related, quantity. The basic idea is to relate $\overline\phi_k$ to the \emph{leakage} of a well-chosen large cap portfolio. A self-contained derivation is provided in Appendix~\ref{app:phi_derivation}, and Remark~\ref{R_phi_bar_interpretation} below contains a related discussion. Estimates $\widehat \phi_k$ for $\phi_k$ are then obtained based on \eqref{eqn:phi_bar_def} by taking successive differences of $\widehat {\overline \phi}_k$ and using that $\widehat {\overline \phi}_d = 0$.
% \begin{align*}
% \widehat {\overline \phi}_k = \sum_{j=k}^d \widehat \phi_j, \quad k=2,\dots,d, \quad \text{and} \quad \sum_{k=1}^d \widehat \phi_k = 0.
% \end{align*}
Figure~\ref{fig:phi} plots the empirically estimated values for $\widehat {\overline \phi}_k$ and $\widehat \phi_k$. The top panel reveals that the $d$th collision parameter estimate is $\widehat \phi_d = - \widehat{\overline\phi}_{d-1} \approx -0.22$. This is orders of magnitude larger in absolute value than $\widehat \phi_k$ for $k < d$, plotted in the bottom panel of Figure~\ref{fig:phi}. For this reason the bottom panel does not show $\widehat\phi_d$ in order to improve readability. The difference in magnitude is expected as the smallest weight only has a one-sided reflection term in its dynamics, namely when it collides with a larger market weight. The other weights have reflection terms with opposite signs, which contribute whenever they collide with a larger or smaller market weight. These opposing effects partially cancel and produce more moderate values for $\widehat \phi_k$, $k < d$.

%As with the volatility parameters, we first perform isotonic regression of $-\widehat {\overline \phi}_k$ for monotonicity and then apply an exponential smoother to remove estimation noise. The resulting $\widehat \phi_k$ is the obtained by taking successive differences of the vector of smoothed $\widehat {\overline \phi}_k$'s. Note that, in absolute value, $\widehat  \phi_d$ is orders of magnitude larger than the other components of $\widehat \phi$. This is expected as the smallest weight only has a one-sided reflection term in its dynamics; namely when it collides with a larger market weight. Conversely, the other ranked weights have reflection terms with opposite signs contributing whenever they collide with a larger and/or smaller market weight inducing a more moderate value $\widehat \phi_k$.

\begin{figure}
    \centering
    \includegraphics[scale=0.6]{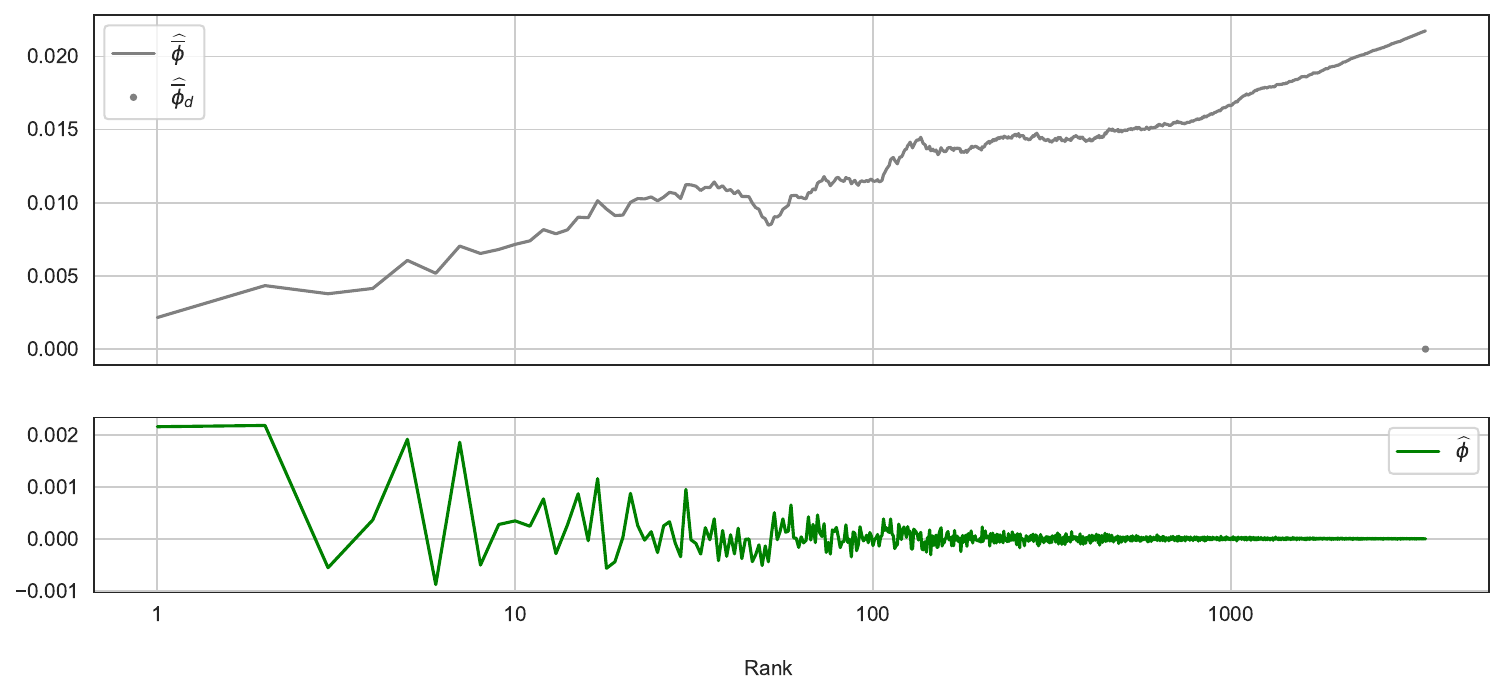}
    \caption{Estimated values for $\widehat {\overline \phi}_k$, $k=1,\ldots,d$ (top panel) and $\widehat \phi_k$, $k=1,\ldots,d-1$ (bottom panel) plotted with the x-axis on a log-scale. The bottom panel does not show $\widehat\phi_d$, which is orders of magnitude larger in absolute value.}
    \label{fig:phi}
\end{figure}

\begin{remark} \label{rem:bar_phi}
The sums $\overline\phi_k$ are positive for $k=1,\ldots,d-1$. Indeed, we have
%$\overline\phi_k = \lim_{T \to \infty} \overline \Phi_{k}(T) / T$ where

%\fbox{how to get the second equality?}\fbox{DI: It follows from \eqref{eqn:Phi} and cancellations of terms when we sum.}
\begin{equation}\label{eqn:bar_phi}
\overline\phi_k = \lim_{T \to \infty} \frac{\overline \Phi_{k}(T)}{T}
\quad \text{where} \quad
\overline \Phi_{k}(T) = \sum_{j=1}^k \Phi_{j}(T) = 
% \sum_{l=1}^{k-1}\sum_{j=1}^k
\sum_{l=1}^{{k}}\sum_{j={k+1}}^{{d}}
\int_0^T\frac{1}{N_l(t)} dL_{l,j}(t).
%\overline \Phi_{k}(T) = \sum_{j=1}^k \Phi_{j}(T) = - \sum_{l=1}^{k-1}\sum_{j=k}^d \int_0^T\frac{1}{N_j(t)} dL_{l,j}(t).
%\overline{\phi}_k := \lim_{T \to \infty} \frac{\overline \Phi_{k}(T)}{T} = \lim_{T \to \infty} \frac{\sum_{j=k}^d \Phi_{j}(T)}{T} = - \sum_{l=1}^{k-1}\sum_{j=k}^d \lim_{T \to \infty} \frac{1}{T}\int_0^T(N_j(t))^{-1} dL_{l,j}(t)
\end{equation}
\end{remark}

\begin{remark} \label{R_phi_bar_interpretation}
One can interpret $\overline\phi_k$ by observing that it measures the asymptotic effect of \emph{leakage} on the growth rate of a certain capitalization-weighted portfolio. Specifically, consider a trading strategy which invests in the top $k$ assets by choosing positions proportionate to the wealth of the investor relative to the market wealth and allocates the remaining holdings to the market portfolio. The portfolio weights of this strategy are
\begin{equation} \label{eqn:exp_portfolio}
\pi_i(t) = X_i(t)\left(\frac{\overline S(t)}{W^\pi(t)}1_{\{r_i(t) \leq k\}} + 1 - \frac{\overline S(t)}{W^\pi(t)}\sum_{j=1}^k X_{(j)}(t)\right)
%    \pi_i(t) =     \underbrace{X_i(t)1_{\{r_i(t) \leq k\}}}_{\text{direct investment in asset } i} + \underbrace{ X_i(t)\left(1 - \sum_{j=1}^k X_{(j)}(t)\right)}_{\text{investment through market portfolio}} 
\end{equation}
for $i= 1,\dots,d$, where $W^\pi$ is the wealth of the investor's portfolio and $\overline S$ is the market wealth. It can be shown that the wealth of this portfolio relative to the market portfolio is given by
\[\frac{W^\pi(T) - \overline S(T)}{\overline S(T)} = \sum_{j=1}^k X_{(j)}(T) - \sum_{j=1}^k X_{(j)}(0) -\overline \Phi_k(T).\]
% \textcolor{blue}{I think I had screwed this up a bit. The portfolio above is generated by the function $G(x) = e^{\sum_{j=1}^k x_{(j)}}$. In view of the minus sign in \eqref{eqn:fg_rank_drift} we should get a minus sign in front of the $\langle X_{(j)}, X_{(l)} \rangle$ terms. Mathematically the rest would be fine but this ruins the interpretation a bit as the portfolio underperforms the market. It would instead make sense then to look at the inverse quantity $\log(\overline S(T)/W^{\pi}(T))$; i.e.\ the performance  of the market relative to this portfolio. In that case the generating function $G(x) = e^{\sum_{j=1}^k -x_{(j)}}$ might be a bit better so that leakage term has the opposite effect of the absolutly continuous drift term. }
% \[ %\begin{equation} \label{eqn:exp_portfolio_growth}
%     \log \left(\frac{W^\pi(T)}{\overline S(T)}\right) = \sum_{j=1}^k X_{(j)}(T) - \sum_{j=1}^k X_{(j)}(0) + \frac{1}{2}\sum_{j,l=1}^k \langle X_{(j)},X_{(l)} \rangle(T) - \overline \Phi_k(T).
% \] %\end{equation}
Dividing by $T$ and sending $T \to \infty$ we obtain that the long-term normalized performance of this portfolio relative to the market is 
\[\lim_{T \to \infty} \frac{1}{T}\frac{W^\pi(T) - \overline S(T)}{\overline S(T)} = - \overline \phi_k.\]
We see that asymptotically the portfolio $\pi$ underperforms the market precisely by the rate $\overline \phi_k$. This effect is known as \emph{leakage}. It arises because the portfolio $\pi$ invests in the top $k$ ranked assets only, and therefore incurs a mechanical rebalancing cost whenever the $k$th and $(k+1)$th ranked weights change place.
%The first term on the right hand side represents the growth rate, in excess of the market, obtained for this portfolio from the \emph{intrinsic volatility} of the market. The second term $\overline \phi_k$ is negative in view of the representation \eqref{eqn:bar_phi} and eats away at the portfolio's growth rate due to \emph{leakage}. Indeed, since the portfolio $\pi$ prescribes investments in the top $k$ ranked assets, there is a mechanical rebalancing cost whenever the $k$ and $(k+1)$st ranked weights change place, which is captured by $\overline \phi_k$. 
\end{remark}

\begin{remark}
Given the explicit representation \eqref{eqn:Phi} of $\Phi$ in terms of semimartingale local times, one may wonder why the estimator \eqref{eqn:hat_bar_phi} is used rather than one based on the occupation density approximation
\[L_{k,j}(t) \approx \frac{1}{\epsilon}\int_0^t1_{[0,\epsilon)}(X_{(k)}(s) - X_{(j)}(s))d\langle X_{(k)} - X_{(j)}\rangle(s)\]
for a small choice of $\epsilon > 0$. The reason is that such an estimator will be data inefficient. Indeed, for small $\epsilon$, where the approximation is theoretically accurate, the integrand may be zero for many observations. In contrast, the estimator in  \eqref{eqn:hat_bar_phi} activates whenever ranks switch, regardless of the resulting distance between the market weights. Moreover, \eqref{eqn:hat_bar_phi} does not depend on a hyperparameter $\epsilon$ that needs to be tuned. 
\end{remark}

\subsection{Growth parameter calibration}  \label{S_growth_calibration}

We are now in a position to obtain estimates for the growth parameters $a_k$ up to one last degree of freedom that we discuss further in the next subsection. The ranked market weight dynamics \eqref{eqn:ranked_weights} and the ergodic property imply the relationship
\begin{equation} \label{eqn:phi_relationship}
    0 = a_k - \lambda \mu_k - \sigma_k^2 \mu_k + \rho_k + \phi_k, \qquad k=1,\dots,d,
\end{equation}
where $\lambda = a_1 + \cdots + a_d$ is the market rate of return (see \eqref{eq_Sbar_returns}), $\phi_k$ is given by \eqref{eq_phi_k_def}, and we set
\begin{equation} \label{eqn:mu_k_rho_k}
\mu_k = \lim_{T\to\infty} \frac{1}{T} \int_0^T X_{(k)}(t) dt, \quad
\rho_k = \lim_{T\to\infty} \frac{1}{T} \int_0^T X_{(k)}(t) \sum_{j=1}^d \sigma_j^2 X_{(j)}(t) dt.
\end{equation}
Thus $\mu_k$ is the long-term average market weight of the $k$th ranked stock, and $\rho_k$ is the long-term average product between the market weight $X_{(k)}$ and the market spot variance $\sum_{j=1}^d \sigma_j^2 X_{(j)}$ (see \eqref{eq_Sbar_returns}). The ergodic property of Theorem~\ref{thm:ergodicity} ensures that these quantities are deterministic and given as the corresponding moments of the invariant measure $\nu$. In \eqref{eqn:phi_relationship} we now replace $\sigma_k^2, \phi_k$ by the estimates $\widehat\sigma_k^2, \widehat\phi_k$, and $\mu_k, \rho_k$ by the estimates
\begin{equation} \label{eqn:mu_hat_rho_hat}
\widehat \mu_k = \frac{1}{N}\sum_{i=0}^{N-1} X_{(k)}(t_i), \quad 
\widehat \rho_k = \frac{1}{N}\sum_{i=0}^{N-1} X_{(k)}(t_i) \sum_{j=1}^d \widehat \sigma_j^2 X_{(j)}(t_i).
\end{equation}
This leads to the growth parameter estimates
\begin{equation} \label{eqn:widehat_a}
    \widehat a_k =  \lambda \widehat \mu_k + \widehat \sigma_k^2 \widehat \mu_k - \widehat \rho_k - \widehat \phi_k, \qquad k=1,\dots,d.
\end{equation}
Summing over $k$ and using that $\sum_{k=1}^d \widehat\mu_k = 1$, $\sum_{k=1}^d \widehat\rho_k = \sum_{k=1}^d \widehat\sigma_k^2 \widehat\mu_k$, and $\sum_{k=1}^d \widehat\phi_k = 0$, we find that $\lambda$ equals $\sum_{k=1}^d \widehat a_k$ as it should. We conclude that the growth parameters are pinned down up to the choice of the market return parameter $\lambda$. Figure~\ref{fig:a_comparison} shows the resulting estimates for $\lambda \in \{0, 0.11, 0.2\}$. These choices are explained in Subsection~\ref{S_return_param} below, where we discuss how to choose $\lambda$.

\begin{figure}
    \centering
    \includegraphics[scale=0.5]{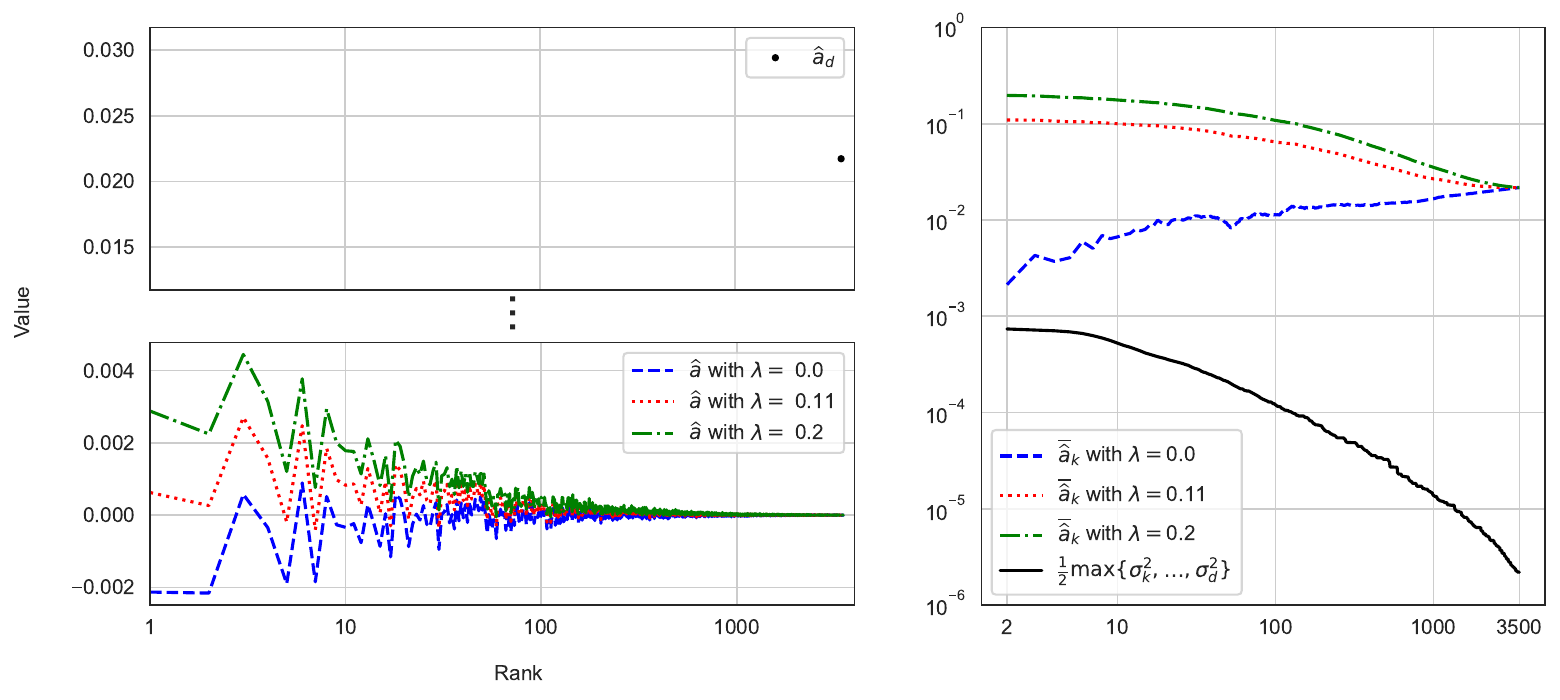}
    \caption{The left panel depicts $\widehat a$ for different choices of $\lambda$ while the right plot compares the tail sum parameter vector $\overline{\widehat a}_k$ to the tail maximum parameter $\frac{1}{2}\max\{\sigma_k^2,\dots,\sigma_d^2\}$ on a log-log scale. 
}
    \label{fig:a_comparison}
\end{figure}

For financially relevant choices of $\lambda$, the value of $\widehat a_d$ at the bottom rank $k=d$ turns out to be positive and significantly larger in magnitude than the values $\widehat a_k$ at higher ranks $k < d$. This is a result of the corresponding behavior of the collision parameter $\widehat\phi_d$; see Figure~\ref{fig:phi} and the discussion in Subsection~\ref{S_collision_calibration}. The large size of $\widehat a_d$ is akin to the behaviour observed in the Atlas model and, more generally, in first order models. It can be attributed to the effect of leakage in the market. Namely, because the universe of stocks is ever changing (due to IPO's, splits, mergers, etc.) the smallest stock can be viewed as not only representing its own growth potential, but that of all other equities, both present and future, that are not directly modeled. %\fbox{Adjust and link to previous section:} This interpretation is supported by the link to $\widehat \phi$ given by \eqref{eqn:phi_relationship} and $\widehat \phi_d$ having a large value due to the nonexistence of a $(d+1)$st stock as discussed in the previous subsection.

% In Figure~\ref{fig:a0} we plot the difference $\widehat a_k - \widehat \lambda \widehat \mu_k$ as a function of $k$. This is the part of the growth parameter that does not depend on the choice of $\widehat\lambda$. Note that the value at the bottom rank $k=d$ is significantly larger in magnitude than the values at other ranks, as a result of the corresponding behavior of the collision parameter $\widehat\phi_d$. For financially relevant choices of the market return parameter $\widehat\lambda$, this effect will carry over to the growth parameters, causing $\widehat a_d$ to be much larger than the other $\widehat a_k$. The large size of $\widehat a_d$ is akin to the behaviour observed in the Atlas model and, more generally, in first order models. It can be attributed to the effect of leakage in the market. Namely, because the universe of stocks is ever changing (due to IPO's, splits, mergers, etc.) the smallest stock can be viewed as not only representing its own growth potential, but that of all the other equities, both present and future, that are not directly modeled. \fbox{Adjust and link to previous section:} This interpretation is supported by the link to $\widehat \phi$ given by \eqref{eqn:phi_relationship} and $\widehat \phi_d$ having a large value due to the nonexistence of a $(d+1)$st stock as discussed in the previous subsection.

% \begin{figure}
%     \centering
%     \includegraphics[scale = 0.6]{Plots/a0_plot.pdf}
%     \caption{Estimated values of $\widehat a_k - \widehat \lambda \widehat \mu_k$ as a function of $k$.}
%     \label{fig:a0}
% \end{figure}

\begin{remark} \label{rem:numerical_instability}
The large size of $\widehat a_d$ causes numerical instability for the small market weights when simulating their evolution using an Euler scheme. This is because the smallest market weight can jump several hundred ranks in a single numerical step, even with a small step size. As such, for certain numerical results to come (which we will specify) we only use the top 1000 market weights, which are numerically more stable and, generally speaking, of greater financial interest. Nevertheless, developing more stable numerical schemes for high-dimensional rank-based models is an important direction for future research.
%more interest to practitioners and academics alike. Nevertheless, developing more stable numerical schemes for large dimensional rank-based models is an important direction for future research.
\end{remark}

\subsection{Choosing the market return parameter} \label{S_return_param}

We now have a procedure for calibrating all the parameters with the exception of the market return parameter $\lambda$. We next seek to understand how the choice of $\lambda$ affects the fit to the collisions and capital distribution curve. We thus fix a value of $\lambda$ and consider the rank volatility stabilized model \eqref{eqn:S_dynamics} with parameters
\[
a_k = \widehat a_k, \quad \sigma^2_k = \widehat\sigma^2_k,\quad k=1,\ldots,d,
\]
where $\widehat a_k$ and $\widehat\sigma^2_k$ are obtained from $\lambda$ and the data as described above. Specifically, the $\widehat\sigma_k^2$ are estimated from data as in Subsection~\ref{S_vol_calibration}. We get the $\widehat\phi_k$ as in Subsection~\ref{S_collision_calibration} and $\widehat\mu_k$, $\widehat\rho_k$ from \eqref{eqn:mu_hat_rho_hat}. These are all obtained from the data. Using these quantities together with $\lambda$, we get the $\widehat a_k$ from \eqref{eqn:widehat_a}. We assume here that the global well-posedness condition \eqref{eqn:parameter_condition} is satisfied, which is the case in practice if $\lambda$ is nonnegative; see the right panel of Figure~\ref{fig:a_comparison}.

Using the model defined in this way, we let $\phi_k$ be given by \eqref{eq_phi_k_def} and $\mu_k, \rho_k$ by \eqref{eqn:mu_k_rho_k}. Thus the $\phi_k$ are the theoretical collision parameters and $\mu_k, \rho_k$ the theoretical long-term average market weights and cross-moments, computed in the calibrated model. If the model were calibrated exactly, these would exactly equal $\widehat\phi_k$, $\widehat\mu_k$, $\widehat\rho_k$. Simultaneously fitting all these parameters exactly is however not possible in this model, and our aim is to understand how the discrepancies depend on the choice of $\lambda$.

With the model parameters currently under consideration, the theoretical relationship \eqref{eqn:phi_relationship} states that
\begin{equation} \label{eqn:collision_parameter_model_relation}
    0 = \widehat a_k - \lambda \mu_k - \widehat \sigma_k^2 \mu_k + \rho_k + \phi_k, \qquad k=1,\dots,d.
\end{equation}
Plugging in the defining relation \eqref{eqn:widehat_a} of the growth parameters $\widehat a_k$ and rearranging yields
\[
\phi_k - \widehat\phi_k = (\lambda + \widehat \sigma_k^2) \left( \mu_k - \widehat\mu_k \right) - \left(\rho_k - \widehat \rho_k \right).
\]
These quantities have different orders of magnitudes across rank, so we normalize by $\widehat \mu_k$ to obtain
\begin{equation} \label{eqn:rel_phi_diff}
    \frac{\phi_k - \widehat \phi_k}{\widehat \mu_k} = (\lambda + \widehat \sigma_k^2)\frac{\mu_k - \widehat \mu_k}{\widehat \mu_k} - \frac{\rho_k - \widehat \rho_k}{\widehat \mu_k}.
\end{equation}
Empirically the largest market weight $X_{(1)}(t)$ is of the order $10^{-1}$ and the largest volatility parameter $\sigma_1^2 = \widehat\sigma_1^2$ is of the order $10^{-2}$ or less. The quantities corresponding to lower ranks quickly become much smaller still. Thus, in view of \eqref{eqn:mu_k_rho_k} and \eqref{eqn:mu_hat_rho_hat}, $\rho_k$ and $\widehat\rho_k$ are several orders of magnitude smaller than $\widehat\mu_k$. This renders the second term on the right-hand side of \eqref{eqn:rel_phi_diff} negligible. If moreover $\lambda$ is of a larger order of magnitude than $\widehat\sigma_k^2$, which is the case for most (but not all) values of $\lambda$ we consider, then we deduce the approximate relationship 
\begin{equation} \label{eqn:approx_relationship}
    \frac{\phi_k - \widehat \phi_k}{\widehat \mu_k} \approx \lambda \left(\frac{\mu_k - \widehat \mu_k}{\widehat \mu_k}\right).
\end{equation}
This suggests that a larger value of $\lambda$ leads either to larger normalized collision errors $(\phi_k - \widehat \phi_k)/ \widehat \mu_k$, or to reduced normalized errors for the average capital distribution curve $(\mu_k - \widehat \mu_k)/\widehat \mu_k$. Empirically we see a combination of both effects. This is illustrated in Figure~\ref{fig:L2}, which displays the $L^2$ (sum-of-squares) error for both quantities as a function of $\lambda$, computed using simulated trajectories of the top 1,000 ranks. That is, for each value of $\lambda$ we plot
\[
\sum_{k=1}^{1000} \left( \frac{\phi_k - \widehat \phi_k}{\widehat \mu_k} \right)^2
\quad\text{and}\quad
\sum_{k=1}^{1000} \left( \frac{\mu_k - \widehat \mu_k}{\widehat \mu_k} \right)^2,
\]
where the theoretical values $\phi_k$ and $\mu_k$, which depend on $\lambda$, are obtained by Monte-Carlo; see Subsection~\ref{sec:cdc} for details. Recall that the empirical estimates $\widehat \phi_k$ and $\widehat\mu_k$ are obtained from data and do not depend on $\lambda$. During the in-sample period the annual growth rate of the entire market was $\lambda \approx 0.11$, indicated with a vertical dashed line in Figure~\ref{fig:L2}. This choice of $\lambda$ offers a balanced compromise between the collision and capital distribution curve fits. As such, this is our preferred choice of $\lambda$, though in the next two sections we further investigate the fit to both the capital distribution curve and collisions and discuss the benefits and drawbacks of choosing a different $\lambda$ parameter. 
% The calibrated growth parameters $\widehat a_k$ for $\lambda \in \{0, 0.11, 0.2\}$, representing small, moderate, and large values of $\lambda$ are plotted in Figure~\ref{fig:a_comparison}. %This is discussed further in Section~\ref{sec:cdc}.

\begin{figure}
    \centering
    \includegraphics[scale = 0.5]{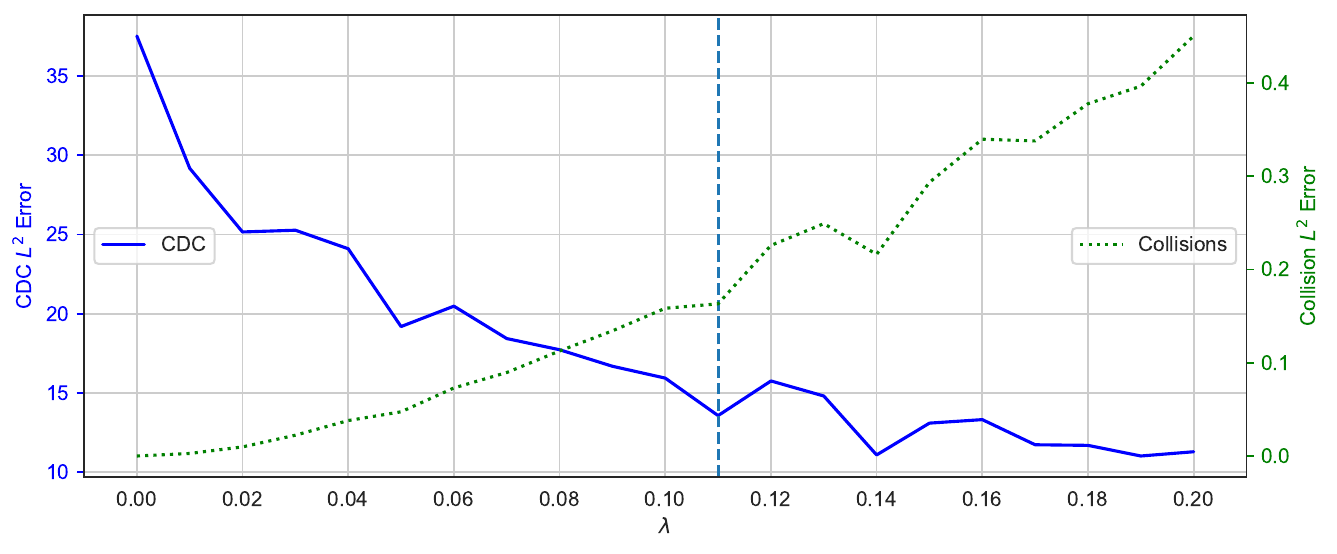}
    \caption{$L^2$ error of simulated capital distributions curve and collisions relative to estimated quantities. Obtained from the top 1000 market weights for 50 simulated samples at each $\lambda$ value. The dashed line is our preferred value of $\lambda = 0.11$.}
    \label{fig:L2}
\end{figure}

\section{Model performance} \label{sec:model_performance}

We now investigate how well the calibrated model fits the data. As described in Subsection~\ref{S_return_param}, we choose a value of the market return parameter $\lambda$ and instantiate the model with the parameters
\[
a_k = \widehat a_k, \quad \sigma^2_k = \widehat\sigma^2_k,\quad k=1,\ldots,d,
\]
where $\widehat a_k$ and $\widehat\sigma^2_k$ are obtained from $\lambda$ and the data using the procedure developed in Section~\ref{sec:calibration}. As discussed there, the data used for the calibration spans the 16-year period from Jan~1, 1990 to Dec~31, 2005. This is the in-sample period. For out-of-sample comparisons we use data from the subsequent 16-year period from Jan~1, 2006 to Dec~31, 2022.

Three different values of the market return parameter are considered, $\lambda \in \{0, 0.11, 0.2\}$, each producing a separate instance of the model. The value $\lambda = 0$ is a lower bound on any plausible average annual rate of return, while we view $\lambda = 0.2$ as a large value, though perhaps realistic in certain years. Our preferred choice $\lambda = 0.11$ is the average annual market rate of return during the in-sample period. We examine the fit to the capital distribution curve (Subsection~\ref{sec:cdc}), the fit to the empirically observed collisions (Subsection~\ref{S_collision_fit}), and the extent to which the model reproduces realistic market weight trajectories (Subsection~\ref{S_trajectories}).

\subsection{Capital distribution curve fit} \label{sec:cdc}

We compare the historical average capital distribution curve with the corresponding theoretical quantities $\mu_k$, $k=1,\ldots,d$, in the calibrated model; see \eqref{eqn:mu_k_rho_k}. Because no formula for $\mu_k$ in terms of the model parameters is available, we use a Monte-Carlo approximation. Specifically, we simulate 50 trajectories up to a large time horizon $T = 100$ years to reach stationarity, inspect the simulated ranked market weights at time $T$, and average across simulation runs to obtain approximations of $\mu_k$, $k=1,\ldots,d$. Next, the historical curve is computed in-sample and out-of-sample. The in-sample curve is given by $\widehat\mu_k$, $k=1,\ldots,d$, in \eqref{eqn:mu_hat_rho_hat} using the same sample that was used to calibrate the model. The out-of-sample curve is computed in the same way, but now using data from the out-of-sample period.

Figure~\ref{fig:cdc} shows the in-sample and out-of-sample historical curves, along with the model generated curves for $\lambda \in \{0, 0.11, 0.2\}$. As suggested by the analysis in Subsection~\ref{S_return_param}, a larger value of $\lambda$ leads to a better fit to the capital distribution curve. Indeed, $\lambda = 0$ performs quite poorly, while $\lambda = 0.11$ and $\lambda = 0.2$ both lead to better fits. The fit is generally best for the middle part of the curve, with moderate deviations in the large stocks. This is expected, since the middle part of the curve is known to be the most stable, while the large stocks have more idiosyncratic fluctuations. Lastly, the small stocks are not well fit. This may be due to model error, but it may also be caused in part by the numerical instability in the simulation of the small stocks discussed in Remark~\ref{rem:numerical_instability}, producing inaccurate values of $\mu_k$ for large $k$. It is not clear which effect dominates.
%It is not clear how large of a role the aforementioned \fbox{where? DI: Remark~\ref{rem:numerical_instability}} numerical instability of the simulations plays and how much is due to model error.
Nonetheless, with $\lambda = 0.11$ and $\lambda = 0.2$, the top 1,000 weights for the simulated capital distribution curves match well with the historical curves, both in-sample and out-of-sample.

\begin{figure}
    \centering
    \includegraphics[scale=0.5]{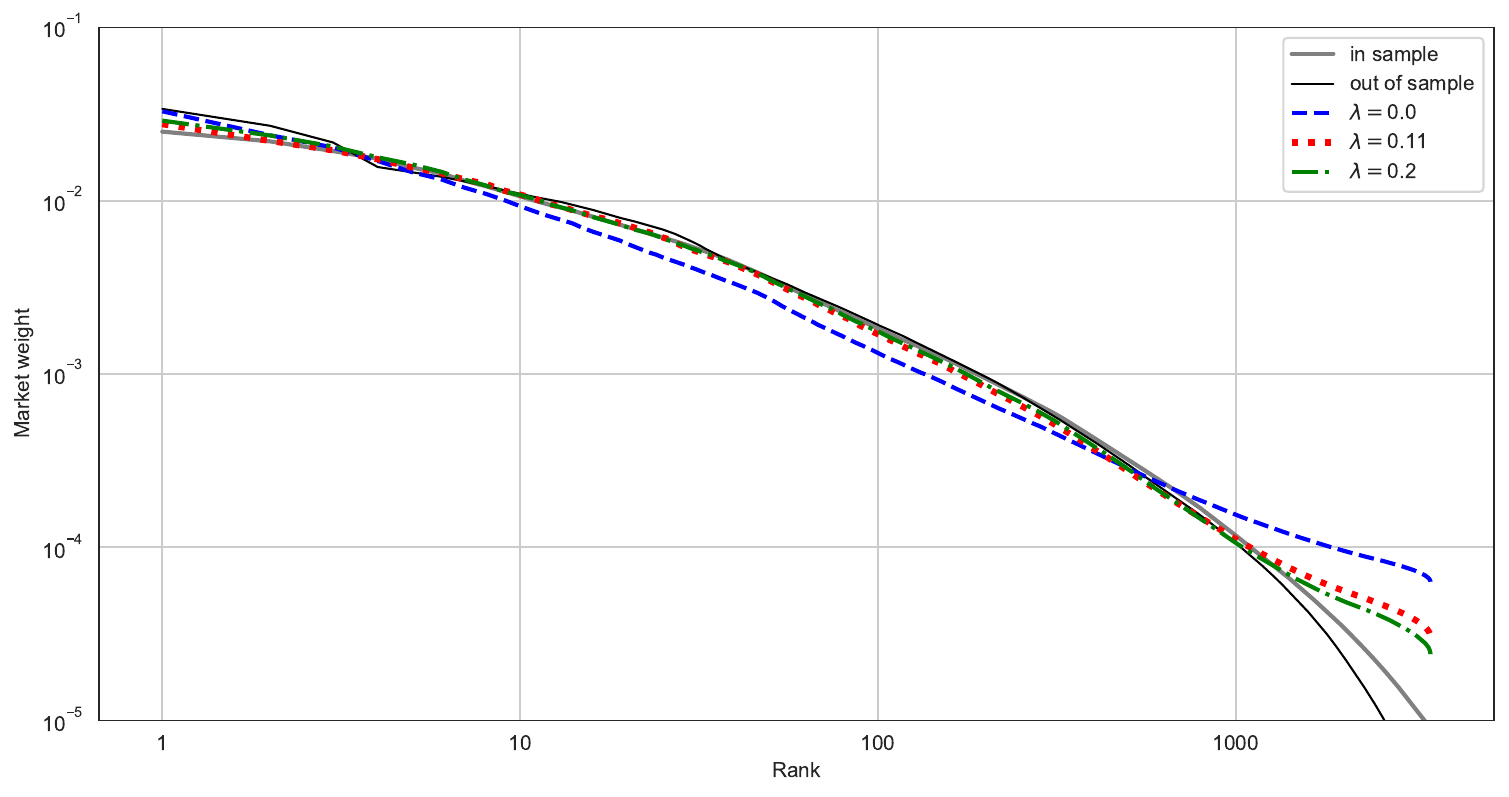}
    \caption{Simulated and historical capital distribution curves.}
    \label{fig:cdc}
\end{figure}

\subsection{Collision fit} \label{S_collision_fit}
We now examine how well the theoretical collision parameters $\phi_k$ (see \eqref{eq_phi_k_def}) in the calibrated model match the empirically estimated values $\widehat\phi_k$. Again no formula for $\phi_k$ in terms of the model parameters is available, so we proceed by Monte-Carlo. Specifically, we use the Monte-Carlo approximations of $\mu_k$ and $\rho_k$ obtained in Subsection~\ref{sec:cdc} together with the relationship \eqref{eqn:collision_parameter_model_relation} to solve for $\phi_k$.
%we first simulate a long \fbox{how long? DI: 100 years, we use the same trajectories as from the CDC part.} trajectory of the full set of market weights. Then, we apply the estimation procedure developed in Subsection~\ref{S_collision_calibration} [\textcolor{blue}{DI: That's not what we do. We use the (previously unlabelled) relationship \eqref{eqn:collision_parameter_model_relation} and our Monte Carlo estimates for $\mu_k$ and $\rho_k$ to obtain the estimates for $\phi_k$. This yields more stable results than using the estimator \eqref{eqn:bar_phi} on the simulated trajectories}] to the simulated data to obtain approximations of $\phi_k$, $k=1,\ldots,d$.
This works well for the large stocks, but becomes inaccurate for smaller stocks due to the numerical instability of the simulation for these stocks discussed in Remark~\ref{rem:numerical_instability}. For this reason we limit the comparison to the top 1,000 ranks.

The normalized in-sample collision errors $(\phi_k - \widehat\phi_k)/\widehat\mu_k$ given by \eqref{eqn:rel_phi_diff} are plotted in Figure~\ref{fig:collisions} for the three instances of the model that we consider, $\lambda \in \{0, 0.11, 0.2\}$.
%As mentioned in Subsection~\ref{S_return_param}, while the orders of magnitude of the collision parameters vary across rank, they are generally the same as for the market weights. For the comparison we thus consider the normalized collisions parameters $\phi_k / \widehat\mu_k$ and $\widehat\phi_k / \widehat\mu_k$. These are plotted for the top 1000 ranks in Figure~\ref{fig:collisions}. Note that the normalized theoretical values $\phi_k / \widehat\mu_k$ are shown for the three instances of the model that we consider, corresponding to $\widehat\lambda \in \{0.001, 0.11, 02\}$.
%\fbox{in-sample vs.\ out-of-sample?}
%\fbox{\textcolor{blue}{We didn't compute this out of sample -- can be done of course.}}
We see that $\lambda = 0$ leads to the best fit, with a near-zero error across all of the plotted ranks. This is consistent with \eqref{eqn:approx_relationship}, which also predicts that the fit becomes worse as $\lambda$ increases. This is borne out in Figure~\ref{fig:collisions}. Additionally, we see that, generally, the error is smaller for the small weights, increases for the middle ranks and is more severe for the largest stocks. Empirically, collisions are more frequent for the small capitalization stocks. For this reason we expect the estimate $\widehat \phi_k$ to be more reliable for large $k$, although confirming this would require an analysis of the statistical properties of the estimator \eqref{eqn:hat_bar_phi}, an interesting topic for future research. Interestingly, it is precisely for larger $k$ that the calibrated models tend to have a better fit.
%stable \fbox{what does this mean?} [\textcolor{blue}{DI: I guess I meant consistent with respect to time and/or needing less data to estimate reliably. Perhaps it is better to say in some elegant way that we expect the estimate $\widehat \phi_k$ to have a lower variance that $\widehat \phi_j$ whenever $k > j$ (i.e.\ it converges faster)}] than the estimate $\widehat \phi_j$ whenever $k > j$. As such, the smaller normalized collision errors $(\phi_k - \widehat \phi_k)/\widehat \mu_k$ for large $k$ observed in Figure~\ref{fig:collisions} indicate that the calibrated model has a better fit to the collision parameters that are most empirically stable.

The normalized out-of-sample collision errors are shown in Figure~\ref{fig:collisions_out_sample}. As there is no analog of \eqref{eqn:rel_phi_diff} for the out-of-sample errors, we compute them by adding the difference between the in- and out-of-sample collision parameters to the in-sample errors. We infer from Figure~\ref{fig:collisions_out_sample} that the normalized errors are larger out-of-sample than in-sample. The difference in the appearance of Figure~\ref{fig:collisions} and Figure~\ref{fig:collisions_out_sample} may be due in part to overfitting to the in-sample data, but is likely also affected by noise in the estimation of the collision parameters. A further analysis of this point, including regularization procedures to reduce estimation noise, is an interesting question that is left for future work.

\begin{figure}
    \centering
    \includegraphics[scale = 0.5]{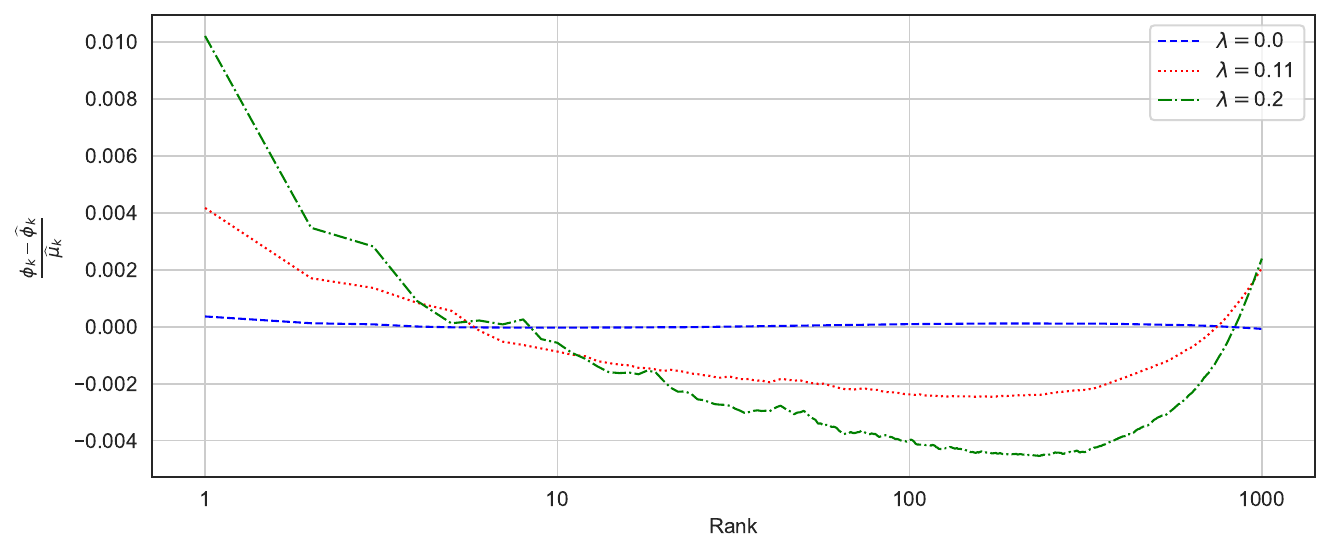}
    \caption{The normalized in-sample collision errors $(\phi_k - \widehat\phi_k)/\widehat\mu_k$ as a function of rank for different choices of $\lambda$.}
    \label{fig:collisions}
\end{figure}

\begin{figure}
    \centering
    \includegraphics[scale = 0.5]{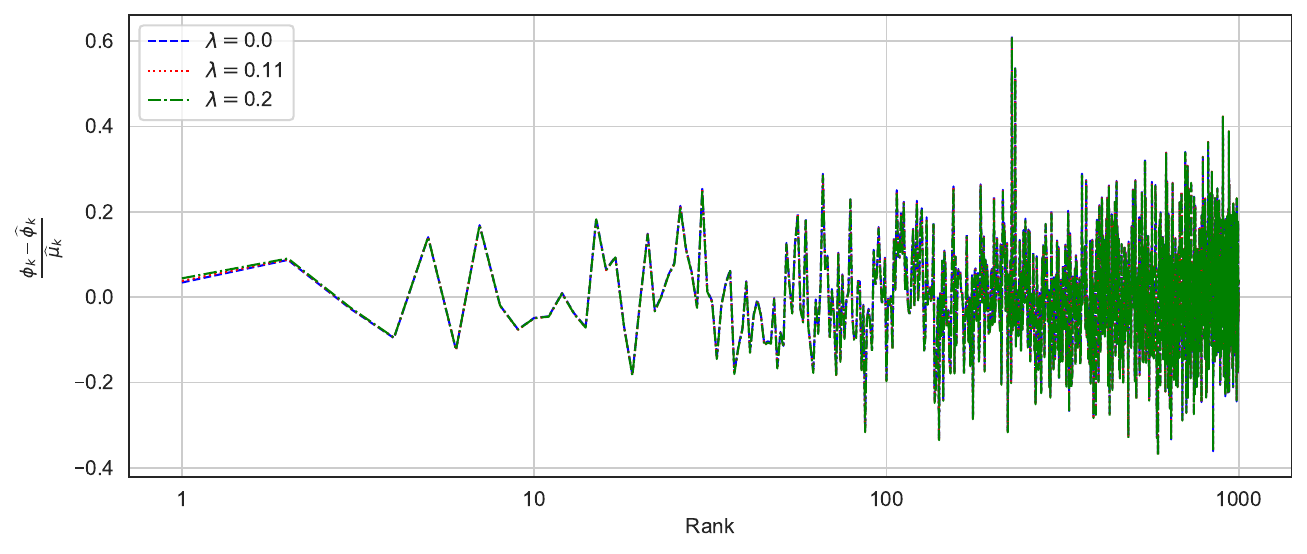}
    \caption{The normalized out-of-sample collision errors as a function of rank for different choices of $\lambda$.}
    \label{fig:collisions_out_sample}
\end{figure}

\subsection{Simulated trajectories} \label{S_trajectories}

Here we compare simulated trajectories from the models with different choices of $\lambda$. Figure~\ref{fig:trajectories} plots trajectories of $X_{(k)}$ for $k = 1,10,50,100,500$ and $1000$ produced from the calibrated model and compared with the historical trajectories both in-sample and out-of-sample. For the in-sample and out-of-sample periods, the models were initiated at the historical values on Jan 2, 1990 and Jan 3, 2006 respectively and simulated for the length of the respective period. At higher ranks $k=1,10,50,100$, the model trajectories are not immediately distinguishable to the eye from the historical trajectory, although they do, of course, have different capital distribution curve and collision profiles as Figures~\ref{fig:cdc} and \ref{fig:collisions} show. This is true both in-sample and out-of-sample. At lower ranks $k=500, 1000$, the historical trajectory is more clearly distinguishable. Furthermore, out-of-sample the model and historical trajectories trend toward different values over time, in contrast to the in-sample behaviour. Although the general shape and structure of the capital distribution curve is stable over time, the empirical curves for the in-sample and out-of-sample periods do differ quantitatively, which causes this discrepancy. Nevertheless, the simulated trajectories are relatively inexpensive to produce and can serve as helpful synthetic data for practitioners, academics and regulators alike.

%\fbox{to do} Here we compare simulated trajectories from the models with different choices of $\lambda$. Figure~\ref{fig:trajectories} plots trajectories of the $X_{(k)}$ for $k = 1,10,50,100,500$ and $1000$ produced from the calibrated model and compared with the historical trajectories both in and out of sample. For the in sample and out of sample periods, the models were initiated at the historical values on Jan 2, 1990 and Jan 3, 2006 respectively and simulated for the length of the respective period. We see that across the board the blue curve representing the $\lambda = 0.001$ model deviates from the others, which is expected in view of its poorly matched capital distribution curve observed in Figure~\ref{fig:cdc}. The $\lambda = 0.11$ and $\lambda = 0.2$ trajectories do not have features that are easily distinguishable to the eye though they do, of course, have different capital distribution curve and collision profiles as Figures~\ref{fig:cdc} and \ref{fig:collisions} show. The out of sample trajectories, however, trend toward different values over time in comparison to the historical trajectory, which is in contrast to the behaviour exhibited for the in sample trajectories. Although the general shape and structure of the capital distribution curve is stable over time, the empirical curves for the in sample and out of sample periods do differ quantitatively, which causes this discrepancy. Nevertheless, the simulated trajectories are relatively inexpensive to produce and can serve as helpful synthetic data for practitioners, academics and regulators alike.

\begin{figure}
    \centering
    \includegraphics[scale=0.5]{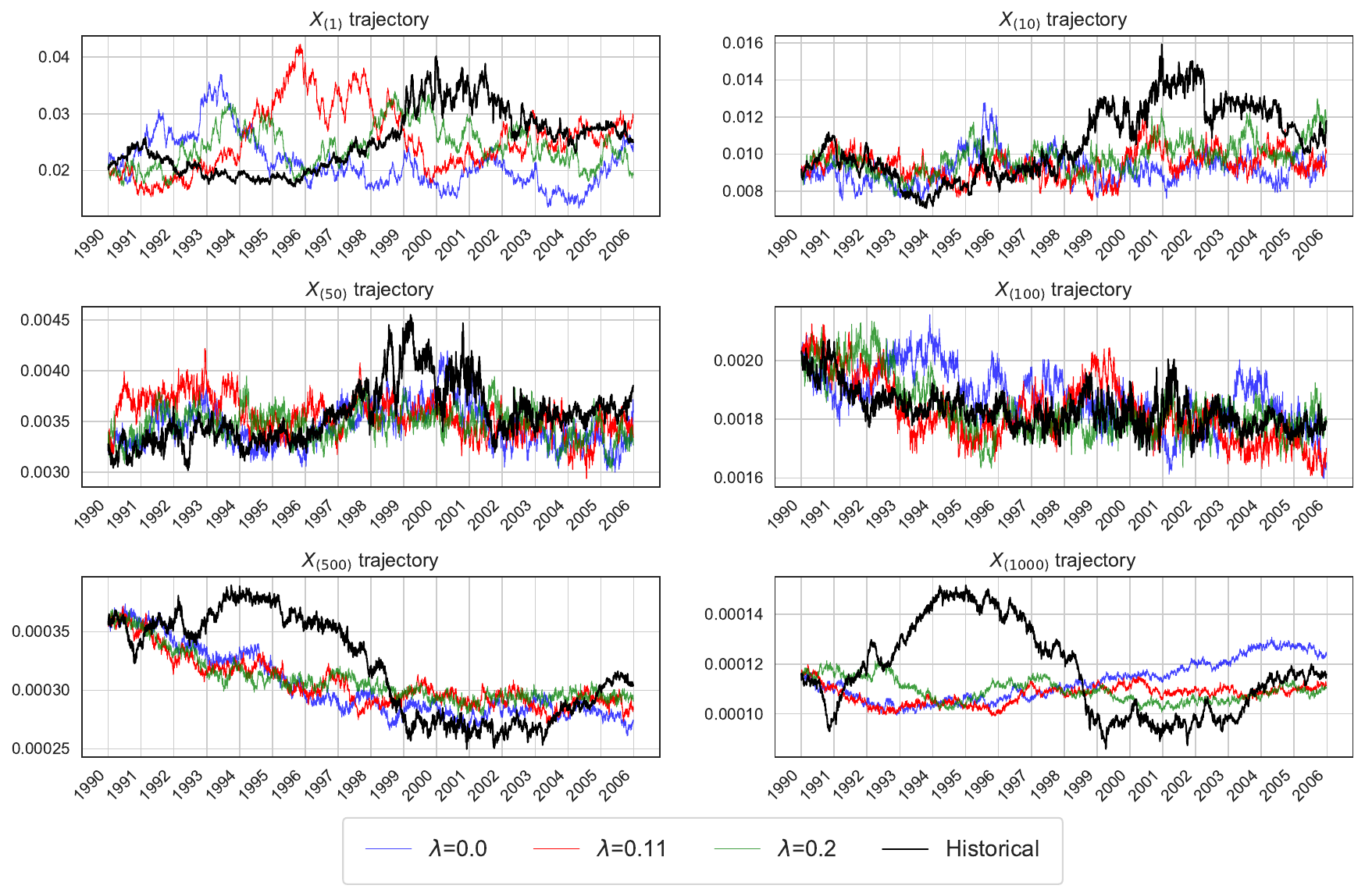}
    \includegraphics[scale=0.5]{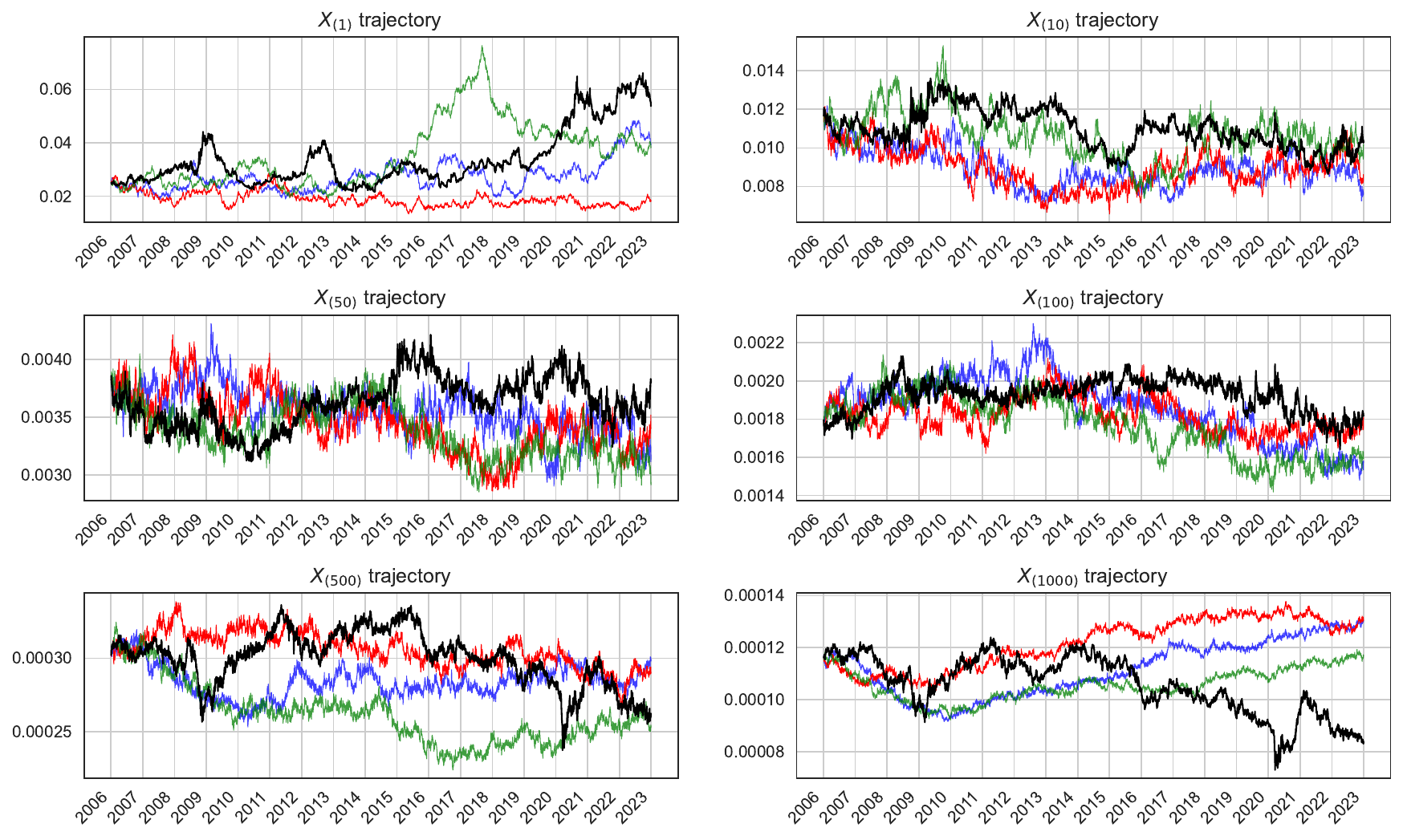}
 
    \caption{Sample path trajectories of the calibrated model with different choices of $\lambda$ together with the historical trajectories in-sample (top panels) and out-of-sample (bottom panels).}
    \label{fig:trajectories}
\end{figure}

\subsection{Summary} \label{sec:model_performance_summary}

Figures~\ref{fig:cdc} and \ref{fig:collisions} clearly show that a small value of $\lambda$ provides a good fit to the collisions and poor fit to the capital distribution curve, while the opposite is true for large values of $\lambda$. Our preferred choice $\lambda = 0.11$ balances these effects while also matching the in-sample market rate of return. However, a user of the model who requires a tighter fit to the capital distribution curve may wish to use a larger value of $\lambda$, while one who is more interested in matching the market turnover as measured by collisions may prefer a smaller value of $\lambda$. Moreover, Figure~\ref{fig:trajectories} indicates that the behaviour of individual trajectories is not strongly dependent on the choice of $\lambda$. Nonetheless, if the model is used to generate synthetic data, it may be useful to use a range of values of $\lambda$ for robustness and in order to capture different growth regimes.

%\fbox{Use this in Section 4?}\fbox{\textcolor{blue}{Good with me.}}
%By tuning $\lambda$ we are able find a balanced fit to both the empirical collisions and capital distribution curves (see Figures~\ref{fig:cdc} and \ref{fig:collisions} below). Interestingly, the value of $\lambda$ fitting the empirical annual rate of return for the entire market lies in the range which offers a good balance between the competing collision and capital distribution curve fits (see Figure~\ref{fig:L2} below). We view this value for the final parameter as our preferred calibrated model, though we also discuss some benefits and drawbacks for taking a smaller/larger value for $\lambda$. 
The ability for the one final parameter $\lambda$ to ensure a good fit for the collisions, capital distribution curve and market rate of return is remarkable. Indeed, the capital distribution curve 
%plots $k \mapsto X_{(k)}$ on a log-log scale for $k =1,\dots,d$, which is known to be stable over time. It 
is a high dimensional object, and all but one parameter were calibrated using criteria not directly related to the curve. Our findings support the analysis of Fernholz \cite{fernholz2002stochastic} done with first order models and provides additional evidence that there is an inherent low-dimensional relationship between rank-based volatility, turnover in the market, the rate of return of the entire market and the distribution of capital.

\section{A discussion of portfolio performance}  \label{sec:portfolio_main}

In this section we analyze the behavior of certain portfolios in the rank volatility stabilized model. Section~\ref{sec:portfolio_def} introduces trading strategies, wealth processes and functional generation of portfolios.   In Section~\ref{sec:growth_optimal} we derive \emph{explicit} expressions for the growth-optimal portfolio in both the closed and open market (precise definitions below). As is often the case, growth-optimal portfolios are highly leveraged and the case studied here is no exception. The open market setup, however, is able to reduce some of the leverage but does not eliminate it entirely. Additionally, growth optimal portfolios provide insight on the mechanics of growth maximization and can serve as a useful starting point when incorporating other features into portfolio selection, such as leverage constraints, risk considerations and structural properties. Although detailed implementations of such extensions are outside the scope of this paper, we do discuss these points briefly in Section~\ref{sec:portfolio_discussion}. 

Secondly, as mentioned in the introduction, relative arbitrage is achievable in rank-volatility stabilized markets. We establish this in Section~\ref{sec:rel_arb} and show that the \emph{diversity weighted portfolio}, which is a long-only functionally generated portfolio, achieves relative arbitrage over a sufficiently long time horizon. The diversity portfolio has empirically been shown to outperform the market portfolio, especially without the presence of transaction costs. Indeed, the recent empirical work \cite{ruf2020impact} reports outperformance of the market portfolio by the diversity portfolio with exponent $p=0.8$ over the period from 1962-2016 with $d=100$ stocks and even with proportional transaction costs of up to 0.5\%. Although possible in the model, we do not claim that probability one outperformance opportunities can be achieved in real equity markets. Indeed, models are not perfect representations of reality and we do not consider any trading frictions here. Nevertheless, it is striking that a model which is able to fit several desirable features of equity markets over long time horizons admits such outperformance.
\subsection{Trading strategies and wealth processes} \label{sec:portfolio_def}

We first review trading strategies and wealth processes. To this end we work on a filtered probability space $(\Omega,\Fcal, (\Fcal_t)_{t \geq 0 }, \P)$ satisfying the usual assumptions and supporting the process $S$ constructed in Theorem~\ref{thm:existence}. Furthermore, we assume that the parameter condition \eqref{eqn:parameter_condition} holds so that we have global existence for $S$ and all assets have strictly positive market capitalizations. An \emph{admissible portfolio} $\pi$ is then a $(\log S)$-integrable process whose components sum to one,
\begin{equation} \label{eqn:portfolio_constraint}
\pi_1(t) + \dots + \pi_d(t) = 1, \qquad t\geq 0.
\end{equation}
The portfolio weight $\pi_i(t)$ represents the proportion of wealth invested in asset $i$. Condition \eqref{eqn:portfolio_constraint} states that the strategy is fully invested; in particular, there is no bank account in this model. We denote the set of all portfolios as $\Pi$ and we say a portfolio $\pi$ is \emph{long-only} if $\pi_i(t) \geq 0$ for every $i=1,\dots,d$ and $t \geq 0$. We always take the initial wealth to be one for simplicity.
 The corresponding wealth process $W^\pi$ is then given by
\begin{equation} \label{eqn:wealth_process_def}
    \frac{dW^\pi(t)}{W^\pi(t)} = \sum_{i=1}^d \pi_i(t) \frac{dS_i(t)}{S_i(t)}, \quad W^\pi(0) = 1.
\end{equation}
Equivalently, we can write $W^\pi$ as a stochastic exponential,
\[W^\pi (T)= \exp\left(\sum_{i=1}^d\int_0^T \frac{\pi_i(t)}{S_i(t)}\, dS_i(t) - \frac{1}{2}\sum_{i,j=1}^d\int_0^T\frac{\pi_i(t)}{S_i(t)}\frac{\pi_j(t)}{S_j(t)}\, d\langle S_i,S_j \rangle (t)\right), \quad T \ge 0.\]
 This representation shows that the wealth process of any admissible portfolio stays strictly positive. In the analysis to come, the logarithmic wealth will play an important role and, as such, we canonically write
\[\log W^\pi(t) = A^\pi(t) + M^\pi(t)\] for the semimartingale decomposition of the log wealth process, where $A^\pi$ is the finite variation part and $M^\pi$ is the local martingale part.

Next we introduce the \emph{market portfolio} $\pi^\mathcal{M}$, which has portfolio weights
\[\pi^{\mathcal{M}}_i(t) := X_i(t), \qquad i =1,\dots,d.\]
To simplify the notation we write $W^{\mathcal{M}}$ for $W^{\pi^{\mathcal{M}}}$ and note from \eqref{eqn:wealth_process_def} that $W^{\mathcal{M}}(t) = \overline S(t) / \overline S(0)$. In the sequel we will use the market portfolio as a benchmark portfolio, which we compare the performance of other portfolios to. Thus, for a portfolio $\pi \in \Pi$ we define the \emph{relative wealth process} of $\pi$ to be 
\[V^\pi(t) := \frac{W^\pi(t)}{W^\mathcal{M}(t)}.\]
Note that, in particular, we have $V^\pi(0) = 1$ for any $\pi \in \Pi$.

Finally, we review the notion of (multiplicatively) functionally generated portfolios. If for some $G:\R^d \to (0,\infty)$ the relative wealth process of a portfolio $\pi \in \Pi$ has the representation
\begin{equation} \label{eqn:master_formula}
    \log V^\pi(T) = \log G(X(T)) - \log G(X(0)) + \Gamma(T)
\end{equation}
for every $T \geq 0$ and some process $\Gamma$ of finite variation with $\Gamma(0) = 0$, then we say that $\pi$ is \emph{functionally generated} by $G$ with \emph{drift process} $\Gamma$. Two important cases that will appear in the sequel are when $G$ is $C^2$, and when it has the representation $G(x) = F(x_{()})$ for some $C^2$ function $F$, where we recall that for any $x \in \R^d$, $x_{()} = (x_{(1)},\ldots,x_{(d)})$ is the vector of its descending order statistics. For more general notions of functionally generated portfolios we refer the reader to \cite{karatzas2017trading}.

\paragraph*{Case 1.} If $G$ is itself $C^2$, the portfolio
\[\pi_i(t) = X_i(t)\left(\partial_i \log G(X(t)) + 1 - \sum_{j=1}^d X_j(t)\partial_j \log G(X(t))\right), \qquad i=1,\dots,d,\] is functionally generated by $G$ and its drift process is 
\[\Gamma(T) = -\int_0^T \sum_{i,j=1}^d \frac{\partial_{ij}G(X(t))}{2G(X(t))} d\langle X_i,X_j \rangle(t).\]

\paragraph*{Case 2.} If $G = F(x_{()})$ for some $C^2$ function $F$, the portfolio characterized by its investment in the ranked assets,
\begin{equation} \label{eqn:fg_rank} \pi_{n_k(t)}(t) = X_{(k)}(t)\left(\partial_k \log F(X_{()}(t)) + 1 - \sum_{k=1}^d X_{(k)}(t)\partial_k \log F(X_{()}(t)) \right),
\end{equation}
is functionally generated by $G$ with drift process
\begin{equation} \label{eqn:fg_rank_drift}
    \Gamma(T) = -\int_0^T \sum_{j,k=1}^d \frac{\partial_{kl}F(X_{()}(t))}{2F(X_{()}(t))} d\langle X_{(j)},X_{(k)} \rangle(t) - \sum_{k=1}^d \int_0^T\partial_k \log F(X_{()}(t)) d\Phi_k(t),
\end{equation} 
where the $\Phi_k$ are the boundary reflection processes in \eqref{eqn:Phi}.

\label{sec:portfolios}
\subsection{Growth-optimal portfolios} \label{sec:growth_optimal}

In this section we study growth-optimal portfolios. Given a collection $\Xi \subset \Pi$ of admissible portfolios we say that a portfolio $\pi^* \in \Xi$ is \emph{growth-optimal} in the class $\Xi$ if $A^{\pi^*} - A^\pi$ is a nondecreasing process for every $\pi \in \Xi$. The existence of a growth-optimal portfolio is known to be equivalent to market viability and implies the absence of arbitrage of the first kind; we refer the reader to the recent monograph \cite{karatzas2021portfolio} for an in-depth analysis of the relationship between these notions.  We now consider two choices for $\Xi$, corresponding to so-called \emph{closed} and \emph{open} markets.

\paragraph*{Closed market.}
%\subsubsection{Closed market} \label{sec:closed_market}

Here we look at the case $\Xi = \Pi$, the full set of admissible portfolios. It\^o's formula implies that for any portfolio $\pi$ we have
\begin{equation} \label{eqn:A_pi}
    A^\pi (T) = \int_0^T \sum_{i=1}^d\left( \frac{a_{r_i(t)}}{X_i(t)}\pi_i(t) - \frac{\sigma^2_{r_i(t)}}{2X_i(t)}\pi_i^2(t)\right)\, dt.
\end{equation}
Thus to obtain the growth-optimal portfolio we simply maximize the integrand pointwise over $\pi$ subject to the constraint \eqref{eqn:portfolio_constraint}. This is a quadratic program with one linear constraint, which has the explicit solution
\[\pi^\mathcal{C}_i(t) = \frac{a_{r_i(t)}}{\sigma^2_{r_i(t)}} - \left(\sum_{j=1}^d \frac{a_j}{\sigma_{j}^2} - 1\right)\frac{X_i(t)/\sigma^2_{r_i(t)}}{\sum_{j=1}^d X_{(j)}(t)/\sigma_j^2}, \qquad i=1,\dots,d.\]
Here the superscript $\Ccal$ indicates optimality in the closed market. A more convenient representation for this portfolio is in terms of its investment in the \emph{ranked} assets,
\[\pi_{n_k(t)}^\Ccal(t) = \frac{a_k}{\sigma_k^2} - \left(\sum_{j=1}^d \frac{a_j}{\sigma_{j}^2} - 1\right)\frac{X_{(k)}(t)/\sigma^2_{k}}{\sum_{j=1}^d X_{(j)}(t)/\sigma_j^2}, \qquad k=1,\dots,d.\]
One way to interpret this portfolio is that the investor seeks to invest the baseline proportion $a_k/\sigma_k^2$ into the $k^{\text{th}}$ largest asset, independent of the asset's market weight. This quantity is downward adjusted by the common factor $\sum_{j=1}^d a_j/\sigma_j^2-1$ weighted against the inverse instantaneous log covariation of the $k^{\text{th}}$ largest asset, 
\[\frac{1}{d\langle \log S_{(k)} \rangle(t)} = \frac{X_{(k)}(t)}{\sigma^2_{k}}\]
relative to the cumulative inverse covariations $\sum_{j=1}^d X_{(j)}(t)/\sigma_j^2$.
% From \eqref{eqn:S_dynamics} we see that this quantity represents the \emph{signal-to-noise ratio} of the $k^{\text{th}}$ largest asset. Typically, these ratios do not lead to a self-financing portfolio. As such, the investor corrects for this by allocating the remainder across all of the assets, with the proportions determined by the inverse of the instantaneous log return of the $k^{\text{th}}$ largest asset.
 Lastly, we note that this portfolio is functionally generated by 
\[G^{\Ccal}(x) :=   \left(\prod_{k=1}^d x_{(k)}^{a_k / \sigma_k^2}\right)\Bigg/ \left(\sum_{k=1}^d \frac{x_{(k)}}{\sigma_k^2}\right)^{\sum_{k=1}^d a_k/\sigma_{k}^2 - 1}.\]

% \textcolor{red}{The functional generation actually implies (via our other results) that this portfolio is asymptotically robust growth-optimal over the class of measures that satisfy exactly what we calibrated to here: volatilities and capital distribution curves. The additional collision requirement can imposed as well, it just restricts the robust class. The ranked-nature (i.e.\ not $C^2$) can be handled as in \cite{itkin2021open}. Moreover, I believe the same applies to the standard open market below. I don't think the same can be said about the relaxed one though because that portfolio doesn't seem to be functionally generated. Shall we put a remark about this in? }

\paragraph*{Open market.} %\subsubsection{Open market} \label{sec:open_market}
The classical closed market setup studied above allows for investment in all of the stocks in the model. However, as noted in the introduction, the constituents making up the $d$ stocks we consider change over time. As such, the closed market setup does not account for turnover in equity markets, which is an important and prevalent feature, especially over long time horizons. The recently proposed framework of \emph{open markets} \cite{fernholz2018numeraire,karatzas2020open}, see also \cite{itkin2021open}, aims to address this issue by only allowing the investor to invest in assets which, at any given time, occupy the top $N$ ranks for a given $N < d$.

In this setup the set of assets available for investment changes over time, as different equities enter and exit the top $N$. Moreover, many investors implicitly or explicitly restrict their investment analysis -- and hence their portfolios -- to a subset of the market. In many cases this subset consists of larger companies, which the open market setup emulates. A  canonical example is an investor who restricts to trading in stocks that make up the S\&P 500 (whose constituents change over time), which the open market with $N = 500$ can serve as a proxy for.

% Recently in \cite{itkin2021open} a \emph{relaxed version} of an open market was proposed where the investor could, in addition to the largest $N$ stocks, invest in the market portfolio. Certain ETF's such as the Wilshire 5000 or Vanguard Total Stock Market ETF can serve as proxies for investments in the market portfolio.  In this section we derive growth-optimal portfolios in both standard and relaxed open markets for rank volatility stabilized models.

% \paragraph{Standard open markets}
Mathematically, investment in the open market of size $N$ is enforced through the admissible portfolio set
\[\Xi = \{\pi \in \Pi: \pi_{n_k(t)}(t) = 0 \text{ for } k> N,\,  t\geq 0\}.\]  We then see from \eqref{eqn:A_pi} that the growth-optimal portfolio in the open market is the pointwise maximizer of 
\[\sum_{k=1}^N \left(\frac{a_k}{X_{(k)}(t)}\pi_{n_k(t)}(t)- \frac{\sigma_{k}^2}{2X_{(k)}(t)}\pi_{n_k(t)}^2(t)\right)\]
subject to the constraint $\pi_{n_1(t)}(t) + \dots + \pi_{n_N(t)}(t) = 1$. This is yet again a linearly constrained quadratic program, and it is straightforward to establish that the solution is
\[\pi^{\Ocal}_{n_k(t)}(t) := \begin{cases} \displaystyle \frac{a_k}{\sigma_k^2} - \left(\sum_{j=1}^N \frac{a_j}{\sigma_{j}^2} - 1\right)\frac{X_{(k)}(t)/\sigma_k^2}{{\sum_{j=1}^N X_{(j)}(t)/\sigma_j^2}} & k=1,\dots,N, \\
0, & \text{otherwise}.
\end{cases}\]
Here the superscript $\mathcal{O}$ indicates optimality in the open market. The portfolio $\pi^{\Ocal}$ has the same structure as $\pi^{\Ccal}$ except it only invests in the largest $N$ securities. The baseline proportion $a_k/\sigma_k^2$ is still present and the allocation of the remainder term is still determined by the inverse of the instantaneous log returns $X_{(k)}(t)/\sigma_k^2$. As in the closed market case, it is easy to verify that this portfolio is functionally generated by
\[G^\Ocal(x)  := \left(\prod_{k=1}^N x_{(k)}^{a_k/\sigma_k^2}\right)\Bigg/ \left(\sum_{k=1}^N \frac{x_{(k)}}{\sigma_k^2}\right)^{\sum_{k=1}^N a_k/\sigma_{k}^2 - 1}.\]

\begin{remark}
A consequence of $\pi^\Ccal$ and $\pi^\Ocal$ being functionally generated is that these strategies have guarantees on their \emph{asymptotic} growth rate under a wider class of market dynamics. Indeed, their asymptotic growth rates remain unchanged under any model where the ranked market weight process $X_{()}$ has the same (i) volatility structure and (ii) invariant measure as in the rank volatility stabilized model. Furthermore, the results of \cite{kardaras2021ergodic,itkin2022ergodic} suggest that these portfolios may have a \emph{robust asymptotic growth-optimality} property over a class of models exhibiting the features (i) and (ii) above and, moreover, that the rank volatility stabilized model serves as a worst-case model in this class. Although the cited papers do not handle the case of non-smooth generating functions, we conjecture that the techniques in \cite{itkin2021open} -- which establishes robust growth optimality in a \emph{relaxed} open market for a nonsmooth rank-based generating function -- may be applicable as the generating functions $G^\Ccal$ and $G^\Ocal$ are of a similar form to the one considered in that paper. 
%A consequence of $\pi^\Ccal$ and $\pi^\Ocal$ being functionally generated is that these strategies have guarantees on their \emph{asymptotic} growth rate under any dynamics of the ranked market weight process $X_{()}$, which have the same (i) volatility structure and (ii) invariant measure as the rank volatility stabilized model. Moreover, the results of \cite{kardaras2021ergodic,itkin2022ergodic} suggest that these portfolios may have an asymptotic robust growth-optimality property over a class of models exhibiting the features (i) and (ii) above and that the rank volatility stabilized model serves as a worst-case measure. Although the cited papers do not handle the case of non-smooth generating functions, the techniques used in \cite{itkin2021open} -- which established robust growth optimality in the case of relaxed open markets for rank Jacobi models -- may be applicable as the generating functions $G^\Ccal$ and $G^\Ocal$ are of a similar form to the one considered in that paper. 
\end{remark}

%\textcolor{blue}{\textbf{A remark on the remark:} in \cite{itkin2021open} we have a non-smooth generating function in an open market setup which is explicit and very similar to $G^{\Ocal}(x)$. I think given the analogous robust class $\Pi$ using the exact same argument we can show that this growth optimal portfolio will have the same growth rate in every model given as some integral against $\nu$. The only difference here is that we know a density for $\nu$ exists, but we don't have more quantitative information about it. The only thing I think we would need to show additionally to get the conjectured result would be that this growth rate expression (i.e.\ integral against the invariant density) is finite and then we would have the result. Intuitively, this should be the case under our parameter condition or possibly some mildly stronger one, but this I believe is the only thing that needs to be addressed. So given this, I think conjecture is ok? It may be that the way it is currently written doesn't stress the similarity of this generating function to the one in \cite{itkin2021open} enough.} 

\begin{remark} \label{rem:lambda_indep}
    At the average capital distribution curve,  when $X_{()}(t) = \widehat \mu$, both portfolios $\pi^\Ccal(t)$ and $\pi^\Ocal(t)$ yield the same value for the calibrated model regardless of the choice of $\lambda$ in \eqref{eqn:widehat_a}.
\end{remark}
\subsection{Relative arbitrage} \label{sec:rel_arb}
The previous section studied portfolios that maximize growth. While being optimal in this sense, such portfolios can be over-leveraged and risky. In this section we instead restrict our attention to long-only functionally generated portfolios and demonstrate that, in the model, such portfolios can outperform the market portfolio with probability one.

\begin{defn}[Relative arbitrage]
    A \emph{relative arbitrage} with respect to the market portfolio on a given time horizon $[0,T^*]$ is a portfolio $\pi \in \Pi$ such that almost surely
\begin{equation}\label{eqn:rel_arb}
    V^\pi(T) > 1  \text{ for all } T \geq T^*.
\end{equation}
\end{defn}
Relative arbitrage is well studied in the literature \cite{fernholz2002stochastic,fernholz2005relative,larsson2021relative} and many sufficient conditions have been derived for its existence. In particular, the paper \cite{banner2008short} establishes instantaneous relative arbitrage in the volatility stabilized market. We now demonstrate that relative arbitrage opportunities exist in the rank volatility stabilized models considered here and show that it can be achieved with the diversity weighted portfolio.

To this end we define the function
\begin{equation} \label{eqn:diversity_function}
D_p(x) = \left(\sum_{i=1}^d x_i^p\right)^{1/p}
\end{equation}
for any $p \in (0,1)$, which generates the diversity weighted portfolio
\begin{equation} \label{eqn:diversity_weights}
    \pi^{D_p}_i(t) := \frac{X_i^p(t)}{\sum_{j=1}^d X_j^p(t)}, \quad i=1,\dots,d.
\end{equation}
Its drift process is $\Gamma(T)= (1-p) \int_0^T \gamma^*(t)\, dt$, where $\gamma^*$ is the \emph{excess growth-rate} given
% \[\gamma^*(t) = \frac{1}{2}\sum_{i=1}^d \pi_i^{D_p}(t) \frac{d \langle \log S_i \rangle_t}{dt} - \sum_{i,j=1}^d\frac{1}{2} \pi_i^{D_p}(t)\pi_j^{D_p}(t) \frac{d \langle \log S_i, log S_j\rangle_t}{dt}\]
in this model by
\begin{equation} \label{eqn:gamma*_equality}
\gamma^*(t)  = \frac{\sum_{k=1}^d X_{(k)}^{p-1}(t)\sigma_k^2}{2\sum_{j=1}^d X_j^p(t)} - \frac{\sum_{k=1}^d X_{(k)}^{2p-1}\sigma_k^2}{2\left(\sum_{j=1}^d X_j^p(t)\right)^2}.
\end{equation}
We look to bound the drift process from below, and start by noting that the second term in the expression for $\gamma^*(t)$ is bounded above by $\frac{1}{2}X_{j^*(t)}(t)^{p-1}\sigma_{j^*(t)}^2 /\sum_{j=1}^d X_{(j)}^p(t)$ where $j^*(t)$ is the index that achieves $\max_{j}\{X_{j}(t)^{p-1}\sigma_j^2\}$. We conclude that
\[
\gamma^*(t) \ge \frac{\sum_{k=1}^d X_{(k)}^{p-1}(t)\sigma_k^2 - X_{j^*(t)}(t)^{p-1}\sigma_{j^*(t)}^2}{2\sum_{j=1}^d X_{(j)}^p(t)} =
\frac{\sum_{k\ne j^*(t)} X_{(k)}^{p-1}(t)\sigma_k^2}{2\sum_{j=1}^d X_{(j)}^p(t)}.
\]
Because $X_{(k)} < 1$ and $p < 1$, the numerator is bounded from below by $\sum_{k \ne j^*(t)} \sigma_k^2 \geq \sum_{k=2}^d \sigma^2_{(k)}$, and the denominator from above by $2 d^{1-p}$. We thus obtain
\begin{equation} \label{eqn:gamma*bound}
    \gamma^*(t) \geq \frac{\sum_{k=2}^d \sigma^2_{(k)}}{2d^{1-p}}.
\end{equation}
% \begin{align} 
% \gamma^*(t)  = \sum_{i=1}^d \left(\frac{1}{2}\pi_i^{D_p}(t)\frac{\sigma_{r_i(t)}^2}{X_i(t)} - \frac{1}{2}\sum_{i=1}^d(\pi_i^{D_p}(t))^2 \frac{\sigma_{r_i(t)}^2}{X_i(t)} \right)
% & =  \frac{\sum_{k=1}^d X_{(k)}^{p-1}(t)\sigma_k^2}{2\sum_{j=1}^d X_j^p(t)} - \frac{\sum_{k=1}^d X_{(k)}^{2p-1}\sigma_k^2}{2\left(\sum_{j=1}^d X_j^p(t)\right)^2}  \label{eqn:gamma*_equality} \\
% & \geq \frac{\sum_{k=1}^d X_{(k)}^{p-1}(t)\sigma_k^2 - \max_{j}\left\{X_{j}(t)^{p-1}\sigma_j^2\right\}}{2\sum_{j=1}^d X_{(j)}^p(t)}  \label{eqn:gamma*_estimate}.
% \end{align}
% If we let $j^*(t)$ be the index that achieves the maximum then we have that the numerator is given by
% \[\sum_{k \ne j^*(t)} X_{(k)}^{p-1}(t)\sigma_k^2 \geq \sum_{k \ne j^*(t)} \sigma_k^2 \geq \sum_{k=2}^d \sigma^2_{(k)},\]
% where we used the fact that $X_{(k)}^{p-1} > 1$ since $p < 1$ in the first inequality and removed the largest component of $\sigma^2$ in the second. Since additionally $\sum_{j=1}^d X_{(j)}^p(t) \leq d^{1-p}$ we obtain from \eqref{eqn:gamma*_estimate} that
% \begin{equation} \label{eqn:gamma*bound}
%     \gamma^*(t) \geq \frac{\sum_{k=2}^d \sigma^2_{(k)}}{2d^{1-p}}.
% \end{equation}
Consequently we see from \eqref{eqn:master_formula} and these bounds that 
\begin{equation} \label{eqn:rel_arb_diversity}
\begin{split}
    \log V^{\pi^{D_p}}(T) &  = \log \left(\frac{D_p(X(T))}{D_p(X(0))}\right) + (1-p)\int_0^T \gamma^*(t)\, dt 
    \geq -\log D_p(X(0)) +  \frac{(1-p)\sum_{k=2}^d \sigma_{(k)}^2}{2d^{1-p}}T,
\end{split}
\end{equation} 
where we also used the fact that $D_p(x) \geq 1$ for every $x \in \Delta^{d-1}_+$. Hence for
\[T > T^*:= \frac{2\log D_p(X(0))d^{1-p}}{(1-p)\sum_{k=2}^d \sigma_{(k)}^2}, \]
the diversity-$p$ portfolio admits relative arbitrage on the time horizon $[0,T]$. Using the calibrated parameter values, $p = 0.8$ and setting $X(0)$ to be the weights on Jan 2, 1990 we have that $T^* \approx 600$. Hence, although theoretically relative arbitrage exists in this model, the period of its realization, even without trading frictions, is far beyond the time horizon of active market participants. It should be noted, however, that the empirical average value of the excess growth-rate $\gamma^*$ computed by evaluating the right hand side of \eqref{eqn:gamma*_equality} along the historical trajectory and averaging over the in sample period is $0.064$. This is similar to the estimated stationary average from the calibrated model with $\lambda = 0.11$, given by $0.061$.
%and is nearly five times larger than the bound on the right hand side of \eqref{eqn:gamma*bound} which evaluates to $0.0145$.
As such, the outperformance time along a typical trajectory is shorter than the model guaranteed time $T^*$.
% For $p =0.8$, the constant on the right-hand side of \eqref{eqn:rel_arb}, which can be viewed as a lower bound on the annual growth rate of $\pi^p$, is given by $0.0029$ for the calibrated parameters. As such, the model implied almost sure lower bound on the relative growth rate of the diversity portfolio is modest

\subsection{Portfolio allocations in the calibrated model} \label{sec:portfolio_discussion}
% \fbox{\textcolor{blue}{Should Figure~\ref{fig:portfolio_weights} show the weights or cumulative weights? Perhaps both with left and right panels?}}

Here we examine typical portfolio weights from the strategies discussed in previous subsections for the calibrated model. Figure~\ref{fig:portfolio_weights} depicts the top $N=100$ portfolios weights when $X_{()}(t) = \widehat \mu$; that is, the market weight vector is at the average capital distribution curve level of the in-sample period. We recall Remark~\ref{rem:lambda_indep}, which notes that for such a market weight configuration the portfolio weights are independent of the choice of $\lambda$. The portfolios $\pi^{\Ccal}$ and $\pi^{D_{0.8}}$ specify nonzero investments in assets with rank $k$ for $k > N$, which are not plotted here, while $\pi^\Ocal$ does not by the open market restriction. We see that the growth-optimal portfolio in the closed market exhibits extreme leverage far beyond any realistically achievable allocation. It prescribes short positions exceeding twice the investor's total wealth in \emph{each} of the 100 largest assets. These short positions are then used to finance an extremely large position in the smallest asset due to its large growth parameter $\widehat a_d$. The portfolio allocations of the growth-optimal portfolio in the open market are not as extreme, but still suffer from infeasible positions. Indeed, the cumulative wealth in the twenty largest stocks is approximately negative four times ones wealth, while the cumulative investment in the top fifty assets reaches positive four times ones wealth. In contrast, the diversity portfolio is a long-only strategy with a reasonable allocation rule that even unsophisticated investors could implement.

Growth-optimal portfolios specifying extreme leverage and highly risky positions is a well-documented deficiency of the growth-optimality criterion (see \cite{samuelson1979we} or \cite[Chapter~6.9]{isichenko2021quantitative}). Our recent theoretical study, however, suggested that open markets can help resolve this issue. Indeed, growth-optimal open market portfolios derived in \cite{itkin2021open} for the rank Jacobi models (obtained from rank volatility stabilized models by fixing a constant volatility vector) allowed for admissible parameter specifications leading to \emph{long-only} growth-optimal portfolios in the open market. The empirical analysis here shows that the open market framework does indeed reduce leverage present in the growth optimal portfolio, but, for empirically calibrated parameters, not to a sufficient extent to be directly tradeable.

To conclude this section we refer to several methods, compatible with the rank-based investing structure explored in this paper, to modify the set of admissible portfolios or the optimality criterion in a way that selects portfolios which are practically implementable, while still achieving performance guarantees. This list is not exhaustive and we stress that detailed analyses of these approaches are beyond the scope of the current study.  
    \paragraph{Pathwise Constraints.} One approach is to impose pathwise constraints which limits the leverage or loss that the portfolio exhibits in a pathwise, model-free manner. A common restriction of this type is to impose \emph{long-only} constraints or other, similar, leverage constraints on the portfolio process $\pi(\cdot)$. Analogously, one can impose constraints on the wealth process $V^\pi(\cdot)$ directly, for example via \emph{drawdown constraints} which do not allow the wealth process to fall below a specified fraction of its historical maximum. Growth maximization under drawdown constraints is known to be a tractable problem, where the optimizer is a model-free transformation of the unrestricted growth-optimal portfolio \cite{kardaras2017numeraire}.

    \paragraph{Risk Constraints.} Investors often target a certain realized portfolio volatility or other target risk metrics. Imposing such constraints, either directly or indirectly by penalizing the objective function, can be used to obtain tamer optimal portfolios. Additionally, specifying different optimality criteria and investor preferences, for example via utility functions, is another way to encode risk and leverage considerations directly into the optimization problem. 
    
    %\fbox{\textcolor{blue}{What would be good to cite here? Add anything else to this paragraph?}}
    
    \paragraph{Portfolio Structure Constraints.} Rather than optimizing  over the space of all portfolios one can instead restrict to a subset exhibiting certain desired structural properties. A natural class consists of \emph{functionally generated portfolios} or its subsets, such as those portfolios generated by concave functions. Concave generating functions lead to long-only portfolios, which inherently do not require leverage. For example, the diversity weighted portfolio \eqref{eqn:diversity_weights} is generated by the concave function \eqref{eqn:diversity_function}. Additionally, optimizing over a functionally generated class of portfolios may simplify the optimization problem as the search space is no longer over a high-dimensional vector of portfolio weights, but rather over a single function specifying the investment rule. 
    Approaches of this type have been the central focus of several recent studies. In a previous paper \cite{itkin2020robust} we obtain a partial differential equation characterizing an optimal concave generating function, as well as a numerical method for its solution, in a setting which allows for model uncertainty. In \cite{campbell2021functional} the authors propose and implement an efficient functional portfolio optimization algorithm which is data-driven and based on regularized empirical risk minimization. Very recently, the paper \cite{cuchiero2023signature} considered path functional portfolios based on signature methods which are numerically efficient to implement.
% Although not long-only, the largest short position in the middle panel of Figure~\ref{fig:portfolio_weights} is about -2.16\% of the investors wealth, bringing it within the realm of feasible allocations for sophisticated investors. Additionally, the shape of the growth-optimal open market portfolio weights are quite different from its closed market compatriot indicating a different mechanism for growth maximization in these two markets. In fact, the shape of the portfolio weights $\pi^\Ocal$ resembles the shape of the diversity $p$-portfolio weights (with $p =0.8$) for the largest 20 securities. However, the size of the wealth allocation for these securities is larger for $\pi^{\Ocal}$ than for the diversity portfolio and this is reflected in the differing shape of these two portfolios in the rest of the assets. 

\begin{figure}
    \centering
    \includegraphics[scale = 0.5]{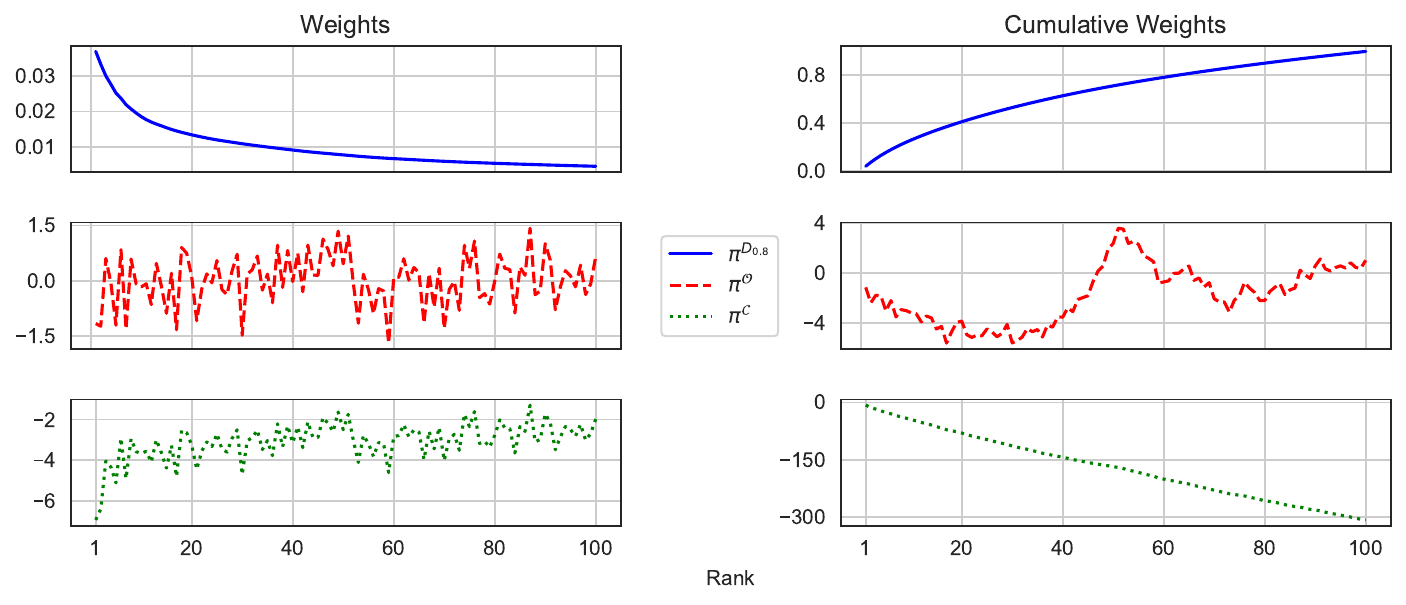}
    \caption{Portfolio weights for the top 100 assets when $X(t) = \widehat \mu$ for the diversity portfolio with $p = 0.8$ (top panels), the growth-optimal portfolio in the open market with $N = 100$ (middle panels) and the growth-optimal portfolio in the closed market (bottom panels). The left panels plot the portfolio weights $\pi_{n_k}$ while the right panels plot the cumulative portfolio weights $\pi_{n_1} + \dots + \pi_{n_k}$. }
    \label{fig:portfolio_weights}
\end{figure}

\appendix

% \subsubsection*{Acknowledgements}
% We thank Johannes Ruf for helpful discussions regarding the data cleaning process and we thank Steven Campbell for pointing out the unbiased estimator \eqref{eqn:sigma_estimator}.             

\section{Proofs} \label{sec:proofs}

In this appendix we prove the theoretical results of Section~\ref{sec:model}. In view of Theorem~\ref{thm:rsde} we also recall the notions of existence and uniqueness for RSDEs for the benefit of the reader.

\subsection{Proof of Theorem~\ref{thm:existence} and Theorem~\ref{prop:no_blowup}} \label{app:pf_existence_no_blowup}

We start with the existence result of Theorem~\ref{thm:existence}. The main tool for proving the result is a classical result of Krylov \cite[Theorem~2.6.1]{krylov2008controlled}, which can be used to establish existence for SDEs with measurable coefficients. However, that result requires uniform ellipticity of the diffusion matrix, which we do not have near the boundary of the domain, so a localization argument is required.

\begin{proof}[Proof of Theorem~\ref{thm:existence}]
	We first establish existence of \eqref{eqn:X_dynamics}. To this end we set $X_d = 1-X_1-\dots - X_{d-1}$ and work with the SDE for $(\widetilde X_1,\dots, \widetilde X_{d-1}) := (X_1,\dots,X_{d-1})$ on 
	\[\widetilde \Delta^{d-1} := \{x \in \R_{+}^{d-1}: x_1 + \dots + x_{d-1} \leq 1\text\},\]
	which is a $(d-1)$-dimensional subset of $\R^{d-1}$. The SDE for $\widetilde X$ is of the form
	\begin{equation} \label{eqn:projected_SDE}
		d\widetilde X_i(t) = \widetilde b_i(\widetilde X(t))\, dt + \sum_{j=1}^{d-1} \widetilde \sigma_{ij}(\widetilde X(t))\, d\widetilde W_j(t), \quad i=1,\dots,d-1
	\end{equation} for coefficients $\widetilde b$ and $\widetilde \sigma$, which can be explicitly worked out from \eqref{eqn:X_dynamics} and a $(d-1)$-dimensional Brownian motion $\widetilde W$. Importantly $\widetilde \sigma$ is uniformly elliptic on $\widetilde \Delta^{d-1}_{+,n} := \widetilde \Delta^{d-1} \cap \{x: \min\{x_{(d-1)}, 1-x_1-\dots-x_{d-1}\} > 1/n\}$ for every $n \in \N$.  Next we extend the diffusion coefficients to all of $\R^{d-1}$ by setting
	\begin{align*}
		\widehat b^n(x)  = \widetilde b(x)1_{\widetilde \Delta^{d-1}_{+,n}}(x), \quad 
		 \widehat \sigma^n(x)  = \widetilde \sigma(x) 1_{\widetilde \Delta^{d-1}_{+,n}}(x) +  I_{d-1}1_{\R^{d-1}\setminus \widetilde \Delta^{d-1}_{+,n}}(x),
	\end{align*}
 where $I_{d-1}$ is the $(d-1)\times (d-1)$ identity matrix.
Since $\widehat\sigma^n$ is uniformly elliptic on $\R^{d-1}$ and both $\widehat b^n$ and $\widehat \sigma^n$ are measurable and bounded, \cite[Theorem~2.6.1]{krylov2008controlled} yields a weak solution $\widehat X^n$ to the SDE 
\[d\widehat X^n(t) = \widehat b(\widehat X^n(t))\, dt + \widehat \sigma(\widehat X^n(t))\, d\widetilde W(t), \quad \widehat X^n(0) = x^0.\]
The stopped process $\widetilde X^n := (\widehat X^n)^{\tau_n}$, where \[\tau_n = \inf\{t \geq 0: \widehat X^n(t) \in \partial \widetilde \Delta^{d-1}_{+,n}\},\] then solves \eqref{eqn:projected_SDE} on the stochastic time interval $[0,\tau_n)$.

Next we note that the laws of $(\widetilde X^n)_{n \in \N}$ are tight since the coefficients $\widehat b^n$ and $\widehat \sigma^n$ are uniformly bounded (see e.g.\ \cite[Theorem~3]{zheng1985tightness}). Consequently, through a subsequence, we may pass to a limiting diffusion $\widetilde X$. Since each $\widetilde X^n$ is a weak solution to \eqref{eqn:projected_SDE}
on $[0,\tau_n)$ it follows that $\widetilde X$ solves \eqref{eqn:projected_SDE} on the stochastic time interval $[0,\widetilde\tau_n)$, where $\widetilde \tau_n = \inf\{t \geq 0: \widetilde X(t) \in \partial \widetilde \Delta^{d-1}_{+,n}\}$. As this holds for every $n$ we see that $\widetilde X$ actually solves \eqref{eqn:projected_SDE} on $[0, \widetilde\tau)$ where 
\[\widetilde \tau = \lim_{n \to \infty} \widetilde \tau_n = \inf\{t \geq 0: \min\{\widetilde X_{(d-1)}(t), 1-\widetilde X_1(t) - \dots - \widetilde X_{d-1}(t)\} = 0\}.\] 
% {\color{red}ML: what kind of limit is $\lim_{t\to\infty} \tau_n$? Note that the $\tau_n$ are potentially defined on different probability spaces. It seems more plausible that a passage to the limit shows that $\widetilde X$ solves \eqref{eqn:projected_SDE} on $[0,\widetilde\tau_n]$ for each $n$, where $\widetilde\tau_n$ is defined as $\tau_n$ but with $\widetilde X$ in place of $\widehat X^n$. This then means that $\widetilde X$ solves \eqref{eqn:projected_SDE} on $[0,\tau)$, where $\tau = \lim_n \widetilde \tau_n$, as desired.}
% \textcolor{blue}{DI: fixed the argument above -- let me know if it looks good now.}

Recalling that $(X_1,\dots,X_{d-1}) = (\widetilde X_1,\dots,\widetilde X_{d-1})$ and $X_d = 1-X_1 - \dots -X_{d-1}$ we see that $\widetilde \tau = \tau^X$ and finally obtain that $X = (X_1,\dots,X_{d})$ solves the original SDE \eqref{eqn:X_dynamics} on the stochastic time interval $[0,\tau^X)$.

To establish existence of \eqref{eqn:S_dynamics} we first consider the above constructed diffusion $X$ with initial condition $x^0 = s^0/(s^0_1 + \dots + s^0_d)$. Next we define
\begin{equation} \label{eqn:Sigma}
	\overline S(t) := (s^0_1 + \dots + s^0_d)\exp\left(\int_0^t\left( \lambda - \frac{1}{2}\sum_{i=1}^d \sigma_{r_i(u)}^2X_i(u) \right)du + \int_0^t \sum_{i=1}^d \sigma_{r_i(u)} \sqrt{X_i(u)}\, dW_i(u)\right),
\end{equation} 
where we recall that $\lambda = a_1 + \dots + a_d$.
Then a direct calculation using the product rule shows that $S(t) := X(t)\overline S(t)$ satisfies \eqref{eqn:S_dynamics} on the the stochastic time interval $[0,\tau^X) = [0,\tau^S)$. 
\end{proof}

Next we establish Theorem~\ref{prop:no_blowup}. The proof follows in a similar fashion to \cite{itkin2021open} by recursively constructing a family of Lyapunov functions.

\begin{proof}[Proof of Theorem~\ref{prop:no_blowup}]
For $k = 1,\dots,d$ we introduce the following notation for this proof
\[\overline X_{(k)} = X_{(k)} + \dots + X_{(d)}, \qquad \overline a_k = a_k + \dots + a_d, \qquad \sigma_k^* = \max\{\sigma_k,\dots,\sigma_d\}.\]

We first prove the result for a process $X$ satisfying \eqref{eqn:X_dynamics}.
Establishing $\tau^X = \infty$ is equivalent to showing that $X_{(d)}$ does not hit zero. We will inductively show that $\overline X_{(k)}$ does not hit zero for $k = 1,\dots,d$ and conclude the result since $\overline X_{(d)} = X_{(d)}$. To this end recall the dynamics of $X_{()}$ given by \eqref{eqn:ranked_weights}. From those dynamics we observe that 
\begin{equation} \label{eqn:barX_k}
\begin{split}
	d\overline X_{(k)}(t) & = \left(\overline a_k - \lambda\overline X_{(k)}(t) - \sum_{j=k}^d\sigma^2_{j}X_{(j)}(t) + \overline X_{(k)}(t)\sum_{j=1}^d \sigma_{j}^2X_{(j)}(t)\right)dt\\
	& \hspace{1cm} + \sum_{j=k}^d\sigma_{j}\sqrt{X_{(j)}(t)}dB_j(t) - \overline X_{(k)}(t)\sum_{j=1}^d \sigma_{j}\sqrt{X_{(j)}(t)}dB_j(t) - d\overline \Phi_k(t),
\end{split}
\end{equation}
    where, as before, $d\overline \Phi_k(t) := \sum_{j=1}^k \Phi_j(t) =  \sum_{l=1}^{k-1}\sum_{j=k}^d (N_j(t))^{-1} dL_{l,j}(t)$.
	Hence, $\overline \Phi_k$ is increasing and
	\begin{equation} \label{eqn:supp_Phi}
		\mathrm{supp}(d\overline \Phi_k) \subset \{t:X_{(k-1)}(t) = X_{(k)}(t)\}.
	\end{equation} 
It will also be useful to define the stopping times
\[
	\tau_k^n := \inf\{t\geq 0: \overline X_{(k)}(t) \leq 1/n\}
\] for any $k \in \{1,\dots,d\}$ and $n \in \N$.

We now proceed by induction. Note that $\overline X_{(1)}(t) = \sum_{i=1}^d X_i(t) = 1$, which never hits zero and establishes the base case. Now pick $k \in \{2,\dots,d\}$ and assume that $\overline X_{(k-1)}(t)$ does not hit zero. We wish to show that $\overline X_{(k)}(t)$ does not hit zero either.  For any $T \geq 0$ and any $m,n \in \N$ we have by It\^o's formula that
\begin{equation} \label{eqn:lyapunov_bound}
\begin{split}
	-\log\left(\frac{ \overline X_{(k)}(T \land \tau_k^n \land \tau_{k-1}^m)}{\overline X_{(k)}(0)}\right) &=  \int_0^{T \land \tau_k^n \land \tau_{k-1}^m}\left[ \frac{1}{2\overline X_{(k)}(t)}\left( \frac{\sum_{j=k}^d\sigma_j^2X_j(t)}{\overline X_{(k)}(t)} - 2\overline a_k - X_{()}(t)^\top \sigma^2\right) +  \lambda\right]dt \\
	& \hspace{0.5cm} + M_k(T \land \tau_k^n \land \tau_{k-1}^m) +\int_0^{T \land \tau_k^n\land \tau_{k-1}^m}\frac{1}{\overline X_{(k)}(t)} d\overline \Phi_k(t) \\
	& \leq \int_0^{T \land \tau_k^n\land \tau_{k-1}^m} \frac{1}{\overline X_{(k)}(t)}\left(\frac{(\sigma_k^*)^2}{2} - \bar a_k\right)\, dt + |\lambda|T  \\
	& \hspace{0.5cm} + M_k(T \land \tau_k^n\land \tau_{k-1}^m) + \int_0^{T \land \tau_k^n\land \tau_{k-1}^m}\frac{2}{\overline X_{(k-1)}(t)} d\overline \Phi_k(t) \\
	& \leq |\lambda|T + M_k(T \land \tau_k^n \land \tau_{k-1}^m) + 2m\overline \Phi_k(T)
	\end{split} 
\end{equation} 
for some martingale $M_k(\cdot \land \tau_k^n \land \tau_{k-1}^m)$. In the first inequality we used \eqref{eqn:supp_Phi}, the fact that $d\overline \Phi_k$ is increasing and the inequality $2\overline X_{(k)}(t) \geq \overline X_{(k-1)}(t)$ whenever $X_{(k-1)}(t) = X_{(k)}(t)$. In the second inequality we used \eqref{eqn:parameter_condition} and the fact that $\overline X_{(k-1)}(t) \geq 1/m$ for $t \leq \tau^m_{k-1}$. Now taking expectation in \eqref{eqn:lyapunov_bound}, sending $n \to \infty$ and using Fatou's lemma we obtain
\[\E\left[-\log\left(\frac{ \overline X_{(k)}(T \land \tau^{\overline X_{(k)}} \land \tau_{k-1}^m)}{\overline X_{(k)}(0)}\right)\right] \leq |\lambda| T + 2m\E[\overline \Phi_k(T)].\]
From the dynamics \eqref{eqn:barX_k} it is clear that $\E[\overline \Phi_k(T)] < \infty$. Since $-\log\overline X_{(k)}(\tau^{\overline X_{(k)}}) = \infty$ it follows that 
\[\P(\tau^{\overline X_{(k)}} < T \land \tau_{k-1}^m) = 0.\]
Since $T$ and $m$ are arbitrary we can send them both to infinity to obtain that 
\[\P(\tau^{\overline X_{(k)}} < \tau^{\overline X_{(k-1)}}) = 0.\]
By the inductive hypothesis we have that $\tau^{\overline X_{(k-1)}} = \infty$, $\P$-a.s.\ so we see that $\tau^{\overline X_{(k)}} = \infty$, $\P$-a.s.\ completing the inductive step. 

The analogous claim about $S$ now follows from the above result for $X$. Indeed, suppose that we start with a diffusion $S$ satisfying \eqref{eqn:S_dynamics}.  Then the total capitalization process $\overline S$ has dynamics
\[
\frac{d\overline S(t)}{\overline S(t)} = \lambda\,  dt + \left( \sum_{i=1}^d \sigma_{r_i(t)}^2 \frac{S_i(t)}{\overline S(t)} \right)^{1/2} d{\overline W}(t) \quad \text{on } [0,\tau^{\overline S})
\]
for some one-dimensional Brownian motion $\overline W$. Thus, $\overline S$ is the stochastic exponential of an It\^o diffusion with uniformly bounded coefficients. Consequently, it cannot hit zero in finite time.
As such, we can define $X_i(t) = S_i(t)/\overline S(t)$ on $[0,\tau^S)$ and note that $S_{(d)}(t) = 0$ if and only if $X_{(d)}(t) = 0$. This is a probability zero event by the result proved for $X$ above, which completes the proof.
\end{proof}

\subsection{RSDEs and the proof of Theorem~\ref{thm:rsde}}
\label{app:rsde}

Next, we turn towards establishing well posedness of the  RSDE that the ranked market weight process \eqref{eqn:ranked_weights} satisfies. First we recall what it means to solve an RSDE. We fix a filtered probability space $(\Omega, \F, (\Fcal(t))_{t\geq0}, \P)$ supporting a Brownian motion $B$ and where $\Fcal(t)$ is the right-continuous enlargement of the filtration generated by $B$. In the definition that follows $D$ is a bounded convex subset of $\R^d$ and $b_i, \sigma_{ij}: \overline D \to \R$ are measurable locally bounded functions.

\begin{defn}
	Let $\tau$ be an $\Fcal(t)$-stopping time.
	A pair of continuous $\Fcal(t)$-adapted processes $(Y,\Phi)$ is a strong solution to the reflected SDE (with normal reflection) 
	\begin{equation} \label{eqn:rsde}
		dY(t) = b(Y(t))\, dt + \sigma(Y(t))\, dB(t) + d\Phi(t), \qquad Y(0) = y^0
	\end{equation}
	on $D$ with initial condition $y^0 \in \overline D$ and on the time interval [0,$\tau$) if
	\begin{itemize}[noitemsep]
		\item $Y(t) \in \overline D$ for every $0 \leq t < \tau$,
		\item $\Phi(t) = (\Phi_1(t),\dots,\Phi_d(t))$ is a finite variation process satisfying for every $0 \leq t < \tau$
  \begin{itemize}[noitemsep]
      \item  $\Phi(0) = 0,$
      \item $\int_0^t 1_{\{Y(s) \not \in \partial D\}}\,d|\Phi(s)| = 0$,
      \item $\Phi(t) = \int_0^t n(Y(s))^\top\, d|\Phi(s)|$ where $n(Y(s))$ is an inward pointing normal vector (uniquely determined almost everywhere with respect to the measure $d|\Phi|$) at $Y(s) \in \partial D$, 
  \end{itemize}
  % {\color{red}ML: This doesn't quite match the definition in Tanaka's paper, see in particular (2.1) which involves a unit vector to the domain. Indeed, when $X(t)$ is in $\partial D$ we want the increment $d\Phi(t)$ to be proportional to an inward unit normal vector at $X(t)$. (In equations with oblique reflection, one would use non-normal vectors instead.)}
  % \textcolor{blue}{DI: thanks for pointing this out, I believe I fixed it now.}
		\item we have $\P$-a.s
		\[Y_i(t) = y^0_i + \int_0^t b_i(Y(t))\, dt + \int_0^t \sum_{j=1}^d \sigma_{ij}(Y(t))\, dB_j(t) + \Phi_i(t), \quad i =1,\dots,d\]
		on the interval $[0,\tau)$.  
	\end{itemize}
We say that pathwise uniqueness holds for \eqref{eqn:rsde} if any two solutions $(Y,\Phi)$ and $(\widetilde Y, \widetilde \Phi)$ are indistinguishable as stochastic processes. 
\end{defn}

We now turn towards proving Theorem~\ref{thm:rsde}. 
\begin{proof}[Proof of Theorem~\ref{thm:rsde}] 

Existence is already established by setting $Y = X_{()}$, where $X$ is the process constructed in Theorem~\ref{thm:existence}. Hence we focus on pathwise uniqueness. To this end we define the subdomains $D_n = \{y \in \nabla^{d-1}_+: y_d > 1/n\}$ and consider the RSDE \eqref{eqn:Y_rsde} on $D_n$. Since the drift and diffusion coefficients are Lipschitz-continuous on $D_n$ we have existence and pathwise uniqueness of \eqref{eqn:Y_rsde} on $D_n$ courtesy of \cite[Theorem~3.1]{tanaka1979stochastic}. Note, however, that if $Y$ is a solution to \eqref{eqn:Y_rsde} on $D$ then the stopped process $Y^{\tau_n}$ is a solution to \eqref{eqn:Y_rsde} on $D_n$ on the time interval $[0,\tau_n)$, where $\tau_n = \inf\{t \geq 0: Y(t) \not \in D_n\}$. By pathwise uniqueness for the RSDE on $D_n$ it then follows that any solution $Y$ to \eqref{eqn:Y_rsde} satisfies $Y = X_{()}$ on $[0,\tau_n)$ since $X_{()}$ is itself a solution to \eqref{eqn:Y_rsde} on the domain $D_n$ and on the time interval $[0,\tau_n)$. Sending $n \to \infty$ yields that $Y = X_{()}$ on the stochastic time interval $[0,\tau^Y) = [0,\tau^{X_{()}}) = [0,\tau^X)$. Under the condition \eqref{eqn:parameter_condition} we know from Proposition~\ref{prop:no_blowup} that $\tau^X = \infty$. Since the solution to \eqref{eqn:Y_rsde} is unique it follows from standard arguments that the solution is a strong Markov process (see e.g.\  \cite[Section~6.2]{Stroock1979Multi}).
\end{proof}

\subsection{Proof of Theorem~\ref{thm:ergodicity}} \label{app:ergodic}

Finally we turn towards proving the ergodic property for $X_{()}$. Existence of an invariant measure is immediate by the compactness of $\nabla^{d-1}$ and non-attainment of the set $\nabla^{d-1}\setminus \nabla^{d-1}_+$.
Ergodicity will then follow from the standard ergodic theory machinery once the uniqueness of an invariant measure is established. The standard approach to obtain this for reflected diffusions was first initiated by Harrison and Williams \cite{harrison1987brownian} for reflected Brownian motion. Our approach follows the same reasoning but is applied to our setting of more general RSDE dynamics.

We first establish the following lemma, which is a version of \cite[Lemma~3.5]{duarte2020reflected} in our setting. Due to the ordered simplex being a $(d-1)$-dimensional subset of $\R^d$, analogously to what was done for $\Delta^{d-1}$ in the proof of Theorem~\ref{thm:existence}, it will be convenient to define its projection onto $\R^{d-1}$. To this end we define the set
\[\widetilde \nabla^{d-1}_+:= \left\{y \in \R^{d-1}_{++}: y_1 \geq \dots \geq y_{d-1} \geq 1- y_1 - \dots - y_{d-1} \text{ and } \sum_{k=1}^{d-1} y_{k} < 1\right\}\] and the projection map $\eta: \nabla^{d-1}_+ \to \widetilde \nabla^{d-1}_+$ via $\eta(y) = (y_1,\dots,y_{d-1})$. We also recall that $\P_y$ denotes the law of $X_{()}$ initiated at $y \in \nabla^{d-1}_+$.

\begin{lem}  \label{lem:technical} 

Assume condition \eqref{eqn:parameter_condition}.
\begin{enumerate}[noitemsep]
\item \label{item:zero_boundary_time} We have  $\int_0^\infty 1_{\partial \nabla^{d-1}_+}(X_{()}(t))\, dt = 0$, $\P_y$-a.s.\ for any $y \in \nabla^{d-1}_+$.
\item \label{item:equivalence} If $\nu$ is an invariant measure for $X_{()}$ then its pushforward $\nu \circ \eta^{-1}$ on $\widetilde \nabla^{d-1}_+$ is equivalent to the Lebesgue measure $\Lcal_{d-1}$ on $\widetilde \nabla^{d-1}_+$. That is
\begin{equation} \label{eqn:equivalence}
    \nu \circ \eta^{-1}(A) = 0 \iff \Lcal_{d-1}(A) = 0
\end{equation}
for every $A \in \mathcal{B}(\widetilde \nabla^{d-1}_+)$.
\end{enumerate}
\end{lem}

% \begin{remark}
% Property \ref{item:null_boundary} of Lemma~\ref{lem:technical} is one that is typically expected to hold for reflected diffusions, but can be quite difficult to establish for general RSDEs. Indeed, an in-expectation version of this theorem  
% \end{remark} 

\begin{proof}
Note that 
\[\partial\nabla^{d-1}_+ = \bigcup_{k=1}^{d-1}\{y \in \nabla^{d-1}: y_k = y_{k+1}\} \cup \{y \in \nabla^{d-1}: y_d = 0\}.\]
Since we already established that $\P_y(X_{(d)}(t) = 0 \text{ for some } t >0 ) = 0$ to prove (i) it just suffices to show that 
\begin{equation} \label{eqn:zero_boundary_time}
    \int_0^\infty 1_{\{0\}}(X_{(k)}(t) - X_{(k+1)}(t))\, dt = 0  
\end{equation} 
for $k = 1,\dots,d-1$. To this end fix $k$ and for $n \in \N$ define the stopping time $\tau_n := \inf\{t \geq 0: X_{(k+1)}(t) \leq 1/n\}$. We also define $\gamma_k(X_{()}(t)) :=  \frac{\langle X_{(k)} - X_{(k+1)} \rangle(t)}{dt}$.  A direction computation using \eqref{eqn:ranked_weights} yields that when $X_{(k)}(t) = X_{(k+1)}(t)$ we have that
% \begin{align*}
%     \eta_k(X_{()}(t)) := &  \frac{\langle X_{(k)} - X_{(k+1)} \rangle(t)}{dt} \\
%     &= \sigma_k^2X_{(k)}(t) + \sigma_{k+1}^2X_{(k+1)}(t) -2 (X_{(k)}(t) - X_{(k+1)}(t))(\sigma_k^2X_{(k)}(t) - \sigma_{k+1}^2X_{(k+1)}(t)) \\ &
%     \hspace{1cm} + (X_{(k)}(t) - X_{(k+1)}(t))^2X_{()}(t)^\top \sigma^2.
% \end{align*}
% Hence, when $X_k(t) = X_{(k+1)}(t)$ we have that 
\[\gamma_k(X_{()}(t)) = (\sigma_k^2 + \sigma_{k+1}^2)X_{(k+1)}(t).\]
For $t \leq \tau_n$ this is bounded from below by $C_n := \frac{\sigma_k^2 + \sigma_{k+1}^2}{n}> 0$. With these preliminary estimates in hand we can now use the occupation density formula to deduce that on the set $\{\omega:  \tau_n(\omega)>T\}$,
\begin{align*} \int_0^T  1_{\{0\}}(X_{(k)}(t) - X_{(k+1)}(t))\, dt &  \leq C_n^{-1} \int_0^T  1_{\{0\}}(X_{(k)}(t) - X_{(k+1)}(t))\gamma_k(X_{()}(t))\, dt \\
 & = C_n^{-1}  \int_0^T  1_{\{0\}}(X_{(k)}(t) - X_{(k+1)}(t)) d\langle X_{(k)} - X_{(k+1)}\rangle(t) \\
 & =  C_n^{-1}\int_{\R} 1_{0}(z) L^z_{X_{(k)} - X_{(k+1)}}(T)\, dz = 0.
\end{align*} 
Since $n$ and $T$ were arbitrary and $\tau_n \to \infty$, $\P_y$-a.s.\ as $n \to \infty$ this establishes \eqref{eqn:zero_boundary_time}. Hence \ref{item:zero_boundary_time} is proved.

Next we turn towards the proof of \ref{item:equivalence}. Since $\Lcal_{d-1}(\partial \widetilde \nabla^{d-1}_+) = 0$ and $X_{()}$ spends zero Lebesgue time on the boundary by part \ref{item:zero_boundary_time} it suffices to prove \eqref{eqn:equivalence} for $A \subset K \subset \widetilde \nabla^{d-1}_+$ where $K$ is a compact set satisfying $\mathrm{dist}(K, \partial \widetilde \nabla^{d-1}_+) > 0$. We fix such an $A$ and following \cite{harrison1987brownian} define 
\begin{align*} \tau & = \inf\{t \geq 0: \eta(X_{()}(t)) \in \partial \widetilde \nabla^{d-1}_+ \text{ or } X_{(d)}(t) \leq 1/m\},
\intertext{and}
\zeta & = \inf\{t \geq 0: \eta(X_{()}(t))\in K \},
\end{align*} where $m$ is chosen so that $\mathrm{dist}(K,\partial \widetilde  \nabla^{d-1}_+) > 1/m$. Next we set $\zeta_0 = 0$ and recursively define
\begin{align*}
\tau_n & := \zeta_{n-1} + \tau \circ \theta_{\zeta_{n-1}}  \\
\zeta_n & := \tau_n + \zeta \circ \theta_{\tau_n} 
\end{align*}
for $n \geq 1$, where $\theta_\cdot$ denotes the shift operator for $X_{()}$. Then we have for any $y \in \nabla^{d-1}_+$ the relationship
\begin{equation} \label{eqn:shift}
    \E_y\left[\int_0^\infty 1_A(\eta(X_{()}(t)))\, dt\right] = \E_y\left[ \sum_{n=1}^\infty \int_{\zeta_{n-1}}^{\tau_n} 1_A(\eta(X_{()}(t)))\, dt\right] =  \E_y\left[ \sum_{n=1}^\infty \E_{\zeta_{n-1}} \left[\int_{0}^{\tau} 1_A(\eta(X_{()}(t)))\, dt\right]\right],
\end{equation}
where, depending on the set $A$, the equality could hold with infinity on both sides.

Note that on the stochastic time interval $[0,\tau)$, $X_{()}$ doesn't hit the boundary of the domain (and consequently neither does the projection $\eta(X_{()}(t))$) so that $\Phi$ in \eqref{eqn:ranked_weights} doesn't contribute to the dynamics on this time interval. Moreover the drift and diffusion coefficients of $\eta(X_{()})$,  which we can write as functions of $\eta(X_{()})$, are Lipschitz continuous on $\widetilde \nabla^{d-1}_{+,m} :=\{y \in \widetilde \nabla^{d-1}_+: 1-y_1 - \dots - y_{d-1} > 1/m\}$. This motivates us to consider an auxiliary SDE on $\R^{d-1}$ 
\begin{equation} \label{eqn:Z_lipschitz}
    dZ(t) = b(Z(t))\, dt + \sigma(Z(t))\, dW(t), \quad Z(0) = z^0,
\end{equation}
where $b:\R^{d-1} \to \R$ and $\sigma: \R^{d-1} \to \mathbb{S}^{d-1}_{++}$ are chosen so that
\begin{itemize}[noitemsep]
    \item $b$ and $\sigma$ are Lipschitz continuous functions with at most linear growth,
    \item $\sigma\sigma^\top$ is uniformly elliptic; i.e. there exists a $\kappa > 0$ such that $\xi^\top \sigma(z)\sigma^\top(z) \xi  \geq \kappa \|\xi\|^2$ for every $\xi, z\in \R^{d-1}$,
    \item $b$ and $\sigma$ coincide with the drift and diffusion coefficients of $\eta(X_{()})$ on the set $\widetilde \nabla^{d-1}_{+,m}$.
\end{itemize} The standard SDE theory yields a pathwise unique solution to \eqref{eqn:Z_lipschitz}. Hence, by localization, we see that almost surely $Z(t) = \eta(X_{()}(t))$ for $t \in [0,\tau)$ whenever $z^0 = \eta(x_{()}^0) \in \widetilde \nabla^{d-1}_{+,m}$. Consequently, we can replace $\eta(X_{()})$ by $Z$ in the right hand side of \eqref{eqn:shift} to obtain
\[  
\E_y\left[\int_0^\infty 1_A(\eta(X_{()}(t)))\, dt\right] = \E_y\left[ \sum_{n=1}^\infty \E_{\zeta_{n-1}} \left[\int_{0}^{\tau} 1_A(Z(t))\, dt\right]\right]. 
\]
Now integrating both sides with respect to $d\nu(y)$ and applying Tonelli's theorem yields
\[
\int_0^\infty \P_\nu(\eta(X_{()}(t))\in A)\, dt = \E_\nu\left[ \sum_{n=1}^\infty \E_{\zeta_{n-1}} \left[\int_{0}^{\tau} 1_A(Z(t))\, dt\right]\right].
\]
We now examine the right hand side. Note that by standard results for well-posed SDEs with Lipschitz continuous coefficients and uniformly elliptic diffusion matrix we have for every $t > 0$ and every initial value $z^0$ that $Z(t)$ has a strictly positive density with respect to the Lebesgue measure. 
% (\textcolor{blue}{What is the best reference for this? The clearest statement I found, in terms of diffusions rather than the associated solutions to the Fokker--Planck equation, was in Stefano De Marco's PhD thesis \cite[Theorem~1.1.1]{de2011probability}. An analytic reference is an old paper of Aronson, which already establishes this result for bounded, measurable, time dependent coefficients (and uniform ellipticity of course) \cite{aronson1967bounds}.)}
It follows that the right hand side is $0$ if $\Lcal_{d-1}(A) = 0$ and infinite otherwise. On the other hand, since $\nu$ is a stationary measure for $X_{()}$ we have that \[\P_\nu(\eta(X_{()}(t))\in A) = \P_{\nu}(\eta(X_{()}(0)) \in A) =  \nu \circ \eta^{-1}(A)\]so that, similarly, it must be that the left hand side is zero if  $\nu \circ \eta^{-1}(A) = 0$ and infinite otherwise. This establishes \ref{item:equivalence} and completes the proof.
\end{proof}

% the associated (non-reflected) SDE  it's dynamics in this time interval are given by
% \[\begin{split}
% 		dX_{(k)}(t) & = \left(\frac{a_k}{2} - \frac{\overline a_1}{2}X_{(k)}(t) - \sigma^2_{k}X_{(k)}(t) + X_{(k)}(t)\sum_{l=1}^d \sigma_{l}^2X_{(l)}(t)\right)dt\\
% 		& \hspace{1cm} + \sigma_{k}\sqrt{X_{(k)}(t)}dB_k(t) - X_{(k)}(t)\sum_{l=1}^d \sigma_{l}\sqrt{X_{(l)}(t)}dB_l(t).
% 	\end{split}  \]
We are now ready to establish Theorem~\ref{thm:ergodicity}.

\begin{proof}[Proof of Theorem~\ref{thm:ergodicity}] As previously mentioned the existence of a stationary distribution is almost immediate. Indeed by viewing $X_{()}$ as a process defined on the state space $\nabla^{d-1}$ we obtain the existence of a stationary probability measure $\nu \in \Pcal(\nabla^{d-1})$ by the Krylov--Bogolyubov theorem (see e.g.\ \cite[Corollary~3.1.2]{da1996ergodicity}). Since, under condition \eqref{eqn:parameter_condition}, $X_{()}$ does not enter the set $\nabla^{d-1} \setminus \nabla^{d-1}_+$ it follows that $\nu(\nabla^{d-1} \setminus \nabla^{d-1}_+) = 0$ so we can view $\nu$ as a probability measure on $\nabla^{d-1}_+$.
Uniqueness follows as a consequence of Lemma~\ref{lem:technical}. Indeed, suppose $\nu_1,\nu_2$ are two stationary measures. Then by Lemma~\ref{lem:technical} their projection to $\R^{d-1}$ are both equivalent to the Lebesgue measure $\Lcal^{d-1}$ and hence equivalent to each other. But by the ergodic decomposition theorem any two distinct invariant measures must have disjoint supports. Hence, it must be that $\nu_1 \circ \eta^{-1}= \nu_2 \circ \eta^{-1}$, which implies that $\nu_1 = \nu_2$ on $\nabla^{d-1}_+$. The remaining claims now follow from standard results in ergodic theory (see e.g.\ \cite{kallen2021found}).

\end{proof}
\section{Rank based estimators} \label{app:ranked_estimators}

In Section~\ref{app:vol_estimator} we compare the performance of the volatility estimator \eqref{eqn:sigma_estimator} with an alternative one using the ranked increments and in Section~\ref{app:phi_derivation} we  derive the collision estimator \eqref{eqn:hat_bar_phi}.

\subsection{Comparison of volatility estimators} \label{app:vol_estimator}

Our volatility parameter estimates are based on the estimator \eqref{eqn:sigma_estimator}, which is a normalized sum of the squared increments $(\log S_{n_k(t_i)}(t_{i+1}) - \log S_{n_k(t_i)}(t_{i}))^2$. As mentioned in Remark~\ref{rem_vol_bias}, one obtains another seemingly natural estimator by using $(\log S_{(k)}(t_{i+1}) - \log S_{(k)}(t_{i}))^2$ instead. While both these estimators are consistent in the fine-discretization limit, the second one is downward biased at fixed discretizations as shown in Figure~\ref{fig:vol-bias}. Although a full analysis is beyond the scope of this paper, let us give a heuristic explanation for the origin of the bias. We wish to acknowledge discussions with Steven Campbell which helped elucidate this point, and indeed alerted us to this issue in the first place. To improve readability, define
\begin{align*}
A &= \log S_{n_k(t_i)}(t_{i+1}) - \log S_{n_k(t_i)}(t_{i}) \\
B &= \log S_{(k)}(t_{i+1}) - \log S_{(k)}(t_{i}) = \log S_{n_k(t_{i+1})}(t_{i+1}) - \log S_{n_k(t_i)}(t_{i}) \\
C &= \log S_{n_k(t_i)}(t_{i+1}) - \log S_{n_k(t_{i+1})}(t_{i+1}).
\end{align*}
The unbiased estimator involves a sum of the $A^2$ and the biased estimator involves a sum of the $B^2$. Note that we have the decomposition
\[
A^2 = (B + C)^2 = B^2 + C^2 + 2BC.
\]
Here the term $BC$ is negligible. Indeed, if the $k$th ranked stock does not switch name from time $t_i$ to $t_{i+1}$, then $C = 0$. Otherwise, the names $n_k(t_i)$ and $n_k(t_{i+1})$ are different and, because the former leaves the $k$th rank and the latter enters the $k$th rank, their respective log-capitalizations will have swapped values with high probability and up to a lower order correction. Thus $B \approx 0$ in this case. Consequently, we have $A^2 \approx B^2 + C^2$, so that the bias originates from the positive term $C^2$. This explains the bias observed in Figure~\ref{fig:vol-bias}. An interesting question for future research is to make the above heuristic reasoning rigorous.

\begin{figure}
    \centering
    \includegraphics[scale = 0.5]{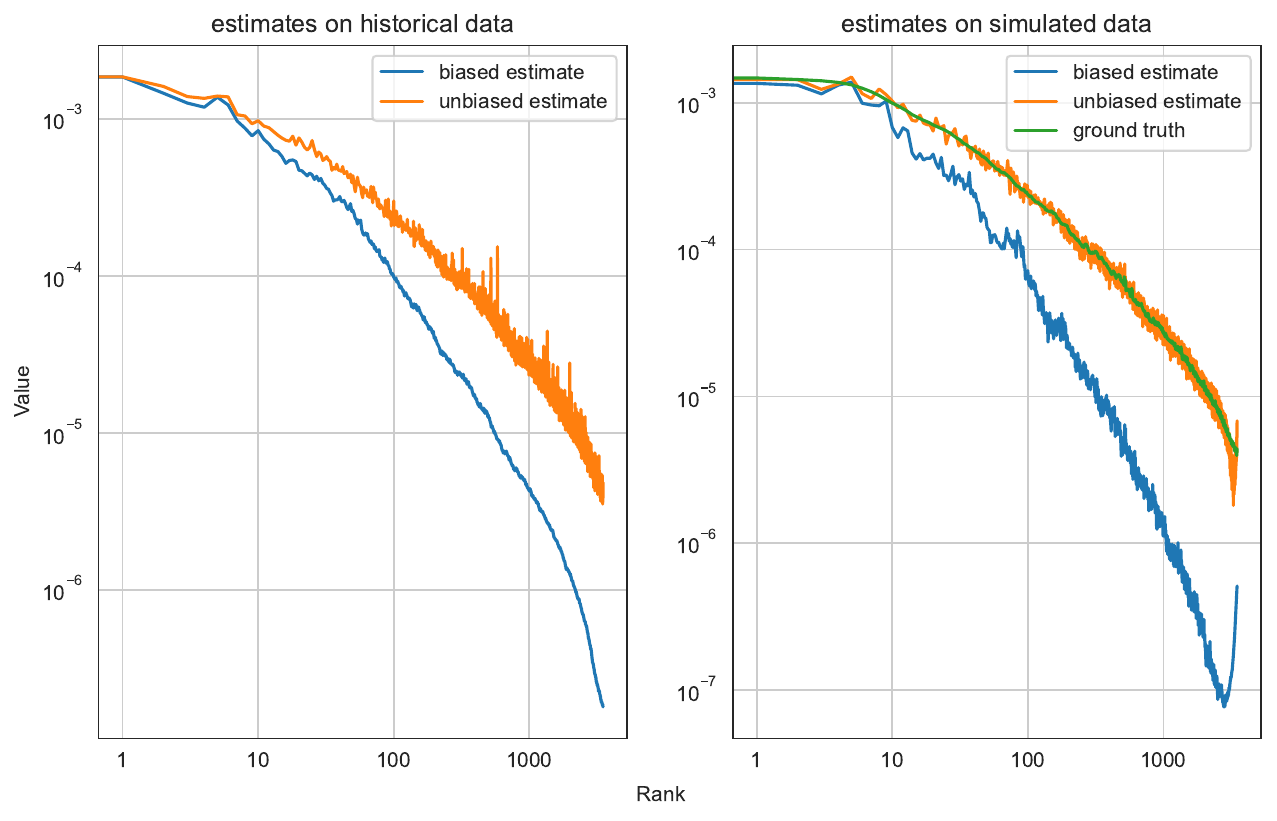}
    \caption{The left panel depicts the raw (unsmoothed) $\sigma^2$ estimates using the two methods on historical data. The right panel shows the performance of the two estimators on a trajectory simulated from the calibrated rank volatility stabilized model sampled at a daily timestep.}
    \label{fig:vol-bias}
\end{figure}

\subsection{Derivation of \eqref{eqn:hat_bar_phi}} \label{app:phi_derivation}
Here we provide a self-contained derivation for the estimator \eqref{eqn:hat_bar_phi} following the approach of Fernholz \cite{fernholz2002stochastic}. To this end fix $k \in \{1,\dots,d-1\}$ and consider the \emph{large cap portfolio} given by
\[\pi^{\mathcal{M}_k}_i(t) :=
  \frac{X_i(t)}{X_{(1)}(t) + \dots + X_{(k)}(t)}1_{\{r_i(t)\leq k\}} \qquad \text{for } i= 1,\dots,d.\]
From this expression it follows from \eqref{eqn:fg_rank} that $\pi^{\mathcal{M}_k}$ is functionally generated by $G(x) = \sum_{j=1}^k x_{(j)}$.  
Writing $V^{\mathcal{M}_k}$ for $V^{\pi^{\mathcal{M}_k}}$ 
we obtain from \eqref{eqn:master_formula} and \eqref{eqn:fg_rank_drift} that 
\begin{align*}
    d\log V^{\mathcal{M}_k}(T) & = d\log G(X(T)) - \frac{1}{X_{(1)}(t) + \dots + X_{(k)}(t)}\, d \overline \Phi_k(t)
\end{align*}
where, as before, $\overline \Phi_k = \Phi_1 + \dots + \Phi_k$.
Hence, we can isolate the reflection term to obtain 
\begin{align} 
\overline \Phi_{k}(T) & = -\int_0^T(X_{(1)}(t) + \dots + X_{(k)}(t))\, d(\log V^{\mathcal{M}_k}- \log(X_{(1)} + \dots + X_{(k)}))(t)  \nonumber \\
& = -\int_0^T(X_{(1)}(t) + \dots + X_{(k)}(t))\, d\log\left(\frac{W^{\mathcal{M}_k}}{S_{(1)}+ \dots + S_{(k)}}\right)(t), \label{eqn:Phi_int}
\end{align}
where in the last equality we recalled that $V^{\mathcal{M}_k}(t) = W^{\mathcal{M}_k}(t)/W^{\mathcal{M}}(t)$, $W^{\mathcal{M}}(t) = \overline S(t)$ and $X_{(j)}(t) = S_{(j)}(t)/\overline S(t)$. 

Now to obtain an estimator for $\overline\phi$ in \eqref{eqn:bar_phi} we seek to discretize the integral in \eqref{eqn:Phi_int}. To this end note that on time periods where the top $k$ constituents in the market do not change, $\pi^{\mathcal{M}_k}$ behaves like a buy-and-hold cap-weighted portfolio in the stocks that, during that time period, are the largest $k$ companies in the market. Rebalancing of the portfolio only occurs when new stocks enter the top $k$. Hence for short observation intervals we can approximate the change in portfolio wealth by the corresponding change in the buy-and-hold portfolio. Indeed, this is precisely the discrete-time implementation of this portfolio. As such, the change in log wealth between observation times $t_i$ and $t_{i+1}$ is approximated by
\begin{equation} \label{eqn:V_discrete}
    \log W^{\mathcal{M}_k}(t_{i+1}) - \log W^{\mathcal{M}_k}(t_i) \approx \log\left(\frac{S_{n_1(t_i)}(t_{i+1}) + \dots + S_{n_k(t_i)}(t_{i+1})}{ S_{n_1(t_i)}(t_i) + \dots + S_{n_k(t_i)}(t_i)}\right).
\end{equation}
Note that the top $k$ stocks at time $t_i$ appear in both expressions, but they are evaluated at different observation times in the numerator and denominator of the right hand side of \eqref{eqn:V_discrete}. Plugging the discretization \eqref{eqn:V_discrete} into \eqref{eqn:Phi_int} yields
\begin{equation} \label{eqn:phi_estimate_final}
\begin{split} \overline \Phi_{k}(T) \approx &  -\sum_{i=0}^{N-1} (X_{(1)}(t_i) + \dots + X_{(k)}(t_i))  \\
& \hspace{0.25cm} \times \left(\log\left(\frac{S_{n_1(t_i)}(t_{i+1}) + \dots + S_{n_k(t_i)}(t_{i+1}))}{S_{(1)}(t_{i+1}) + \dots + S_{(k)}(t_{i+1})}\right) - \log \left(\frac{S_{n_1(t_i)}(t_i) + \dots + S_{n_k(t_i)}(t_i)}{S_{(1)}(t_{i}) + \dots + S_{(k)}(t_{i})}\right)\right) \\
& =\sum_{i=0}^{N-1}(X_{n_1(t_i)}(t_i) + \dots + X_{n_k(t_i)}(t_i))\log\left( \frac{S_{n_1(t_{i+1})}(t_{i+1}) + \dots + S_{n_k(t_{i+1})}(t_{i+1})}{S_{n_1(t_{i})}(t_{i+1}) + \dots + S_{n_k(t_{i})}(t_{i+1})}\right),
\end{split}
\end{equation}
where in the final equality we used the fact that $S_{(j)}(t_i) = S_{n_j(t_i)}(t_i)$ to eliminate one of the log ratio terms. Now dividing by $T$ in \eqref{eqn:phi_estimate_final} yields \eqref{eqn:hat_bar_phi}.

\bibliographystyle{plain}
\bibliography{References}

\begin{thebibliography}{10}

\bibitem{banner2019diversification}
Adrian Banner, Robert Fernholz, Vassilios Papathanakos, Johannes Ruf, and David
  Schofield.
\newblock Diversification, volatility, and surprising alpha.
\newblock {\em Journal of Investment Consulting}, 19(1):23--30, 2019.

\bibitem{banner2008short}
Adrian~D Banner and Daniel Fernholz.
\newblock Short-term relative arbitrage in volatility-stabilized markets.
\newblock {\em Ann. Finance}, 4(4):445--454, 2008.

\bibitem{banner2005atlas}
Adrian~D. Banner, Robert Fernholz, and Ioannis Karatzas.
\newblock Atlas models of equity markets.
\newblock {\em Ann. Appl. Probab.}, 15(4):2296--2330, 2005.

\bibitem{Banner2008Local}
Adrian~D. Banner and Raouf Ghomrasni.
\newblock Local times of ranked continuous semimartingales.
\newblock {\em Stochastic Process. Appl.}, 118(7):1244--1253, 2008.

\bibitem{campbell2022efficient}
Steven Campbell and Ting-Kam~Leonard Wong.
\newblock Efficient convex pca with applications to wasserstein geodesic pca
  and ranked data.
\newblock {\em arXiv preprint arXiv:2211.02990}, 2022.

\bibitem{campbell2021functional}
Steven Campbell and Ting-Kam~Leonard Wong.
\newblock Functional portfolio optimization in stochastic portfolio theory.
\newblock {\em SIAM J. Financial Math.}, 13(2):576--618, 2022.

\bibitem{cuchiero2023signature}
Christa Cuchiero and Janka M{\"o}ller.
\newblock Signature methods in stochastic portfolio theory.
\newblock {\em arXiv preprint arXiv:2310.02322}, 2023.

\bibitem{da1996ergodicity}
G.~Da~Prato and J.~Zabczyk.
\newblock {\em Ergodicity for infinite-dimensional systems}, volume 229 of {\em
  London Mathematical Society Lecture Note Series}.
\newblock Cambridge University Press, Cambridge, 1996.

\bibitem{duarte2020reflected}
Mauricio Duarte.
\newblock Reflected (degenerate) diffusions and stationary measures.
\newblock In {\em X{III} {S}ymposium on {P}robability and {S}tochastic
  {P}rocesses}, volume~75 of {\em Progr. Probab.}, pages 3--35.
  Birkh\"{a}user/Springer, Cham, [2020] \copyright 2020.

\bibitem{fernholz2002stochastic}
E.~Robert Fernholz.
\newblock {\em Stochastic portfolio theory}, volume~48 of {\em Applications of
  Mathematics (New York)}.
\newblock Springer-Verlag, New York, 2002.
\newblock Stochastic Modelling and Applied Probability.

\bibitem{fernholz2018numeraire}
Robert Fernholz.
\newblock Numeraire markets.
\newblock {\em arXiv preprint arXiv:1801.07309}, 2018.

\bibitem{fernholz2013second}
Robert Fernholz, Tomoyuki Ichiba, and Ioannis Karatzas.
\newblock A second-order stock market model.
\newblock {\em Annals of Finance}, 9:439--454, 2013.

\bibitem{fernholz2005relative}
Robert Fernholz and Ioannis Karatzas.
\newblock Relative arbitrage in volatility-stabilized markets.
\newblock {\em Ann. Finance}, 1(2):149--177, 2005.

\bibitem{harrison1987brownian}
J.~M. Harrison and R.~J. Williams.
\newblock Brownian models of open queueing networks with homogeneous customer
  populations.
\newblock {\em Stochastics}, 22(2):77--115, 1987.

\bibitem{ichiba2011hybrid}
Tomoyuki Ichiba, Vassilios Papathanakos, Adrian Banner, Ioannis Karatzas, and
  Robert Fernholz.
\newblock Hybrid atlas models.
\newblock {\em Ann. Appl. Probab.}, 21(2):609--644, 2011.

\bibitem{isichenko2021quantitative}
Michael Isichenko.
\newblock {\em Quantitative portfolio management: The art and science of
  statistical arbitrage}.
\newblock John Wiley \& Sons, 2021.

\bibitem{itkin2022ergodic}
David Itkin, Benedikt Koch, Martin Larsson, and Josef Teichmann.
\newblock Ergodic robust maximization of asymptotic growth under stochastic
  volatility.
\newblock {\em arXiv preprint arXiv:2211.15628}, 2022.

\bibitem{itkin2021open}
David Itkin and Martin Larsson.
\newblock Open markets and hybrid jacobi processes.
\newblock {\em arXiv preprint arXiv:2110.14046}, 2021.

\bibitem{itkin2020robust}
David Itkin and Martin Larsson.
\newblock Robust asymptotic growth in stochastic portfolio theory under
  long-only constraints.
\newblock {\em Math. Finance}, pages 1--58, 2021.

\bibitem{kallen2021found}
Olav Kallenberg.
\newblock {\em Foundations of modern probability}, volume~99 of {\em
  Probability Theory and Stochastic Modelling}.
\newblock Springer Cham, third edition, 2021.

\bibitem{karatzas2021portfolio}
Ioannis Karatzas and Constantinos Kardaras.
\newblock {\em Portfolio Theory and Arbitrage: A Course in Mathematical
  Finance}.
\newblock American Mathematical Society, 2021.

\bibitem{karatzas2020open}
Ioannis Karatzas and Donghan Kim.
\newblock Open markets.
\newblock {\em Math. Finance}, pages 1--52, 2020.

\bibitem{karatzas2017trading}
Ioannis Karatzas and Johannes Ruf.
\newblock Trading strategies generated by {L}yapunov functions.
\newblock {\em Finance Stoch.}, 21(3):753--787, 2017.

\bibitem{kardaras2017numeraire}
Constantinos Kardaras, Jan Ob{\l}{\'o}j, and Eckhard Platen.
\newblock The num{\'e}raire property and long-term growth optimality for
  drawdown-constrained investments.
\newblock {\em Mathematical Finance}, 27(1):68--95, 2017.

\bibitem{kardaras2021ergodic}
Constantinos Kardaras and Scott Robertson.
\newblock Ergodic robust maximization of asymptotic growth.
\newblock {\em Ann. Appl. Probab.}, 31(4):1787--1819, 2021.

\bibitem{krylov2008controlled}
N.V. Krylov.
\newblock {\em Controlled diffusion processes}, volume~14.
\newblock Springer Science \& Business Media, 2008.

\bibitem{larsson2021relative}
Martin Larsson and Johannes Ruf.
\newblock Relative arbitrage: sharp time horizons and motion by curvature.
\newblock {\em Math. Finance}, 31(3):885--906, 2021.

\bibitem{markowitz1952portfolio}
Harry Markowitz.
\newblock Portfolio selection.
\newblock {\em The Journal of Finance}, 7(1):77--91, 1952.

\bibitem{ruf2023github}
Johannes Ruf.
\newblock Empirical finance with equity data ({P}h.{D}. course), {LSE}.
\newblock \url{github.com/johruf/CRSP_on_WRDS_introduction}, 2023.

\bibitem{ruf2020impact}
Johannes Ruf and Kangjianan Xie.
\newblock The impact of proportional transaction costs on systematically
  generated portfolios.
\newblock {\em SIAM J. Financial Math.}, 11(3):881--896, 2020.

\bibitem{samuelson1979we}
Paul~A Samuelson.
\newblock Why we should not make mean log of wealth big though years to act are
  long.
\newblock {\em Journal of Banking \& Finance}, 3(4):305--307, 1979.

\bibitem{Stroock1979Multi}
Daniel~W. Stroock and S.~R.~Srinivasa Varadhan.
\newblock {\em Multidimensional diffusion processes}, volume 233 of {\em
  Grundlehren der Mathematischen Wissenschaften [Fundamental Principles of
  Mathematical Sciences]}.
\newblock Springer-Verlag, Berlin-New York, 1979.

\bibitem{tanaka1979stochastic}
Hiroshi Tanaka.
\newblock Stochastic differential equations with reflecting boundary condition
  in convex regions.
\newblock {\em Hiroshima Math. J.}, 9(1):163--177, 1979.

\bibitem{zheng1985tightness}
Wei~An Zheng.
\newblock Tightness results for laws of diffusion processes application to
  stochastic mechanics.
\newblock In {\em Annales de l'IHP Probabilit{\'e}s et statistiques},
  volume~21, pages 103--124, 1985.

\end{thebibliography}

\end{document}